\documentclass[english,11pt,displaymath,a4paper]{article}
\usepackage[top=2cm,bottom=2cm,left=2cm,right=2.2cm,a4paper]{geometry}
\usepackage{amsfonts}
\usepackage{amsmath}
\usepackage{amssymb}
\usepackage{amsthm}
\usepackage{multirow}
\usepackage{caption}
\usepackage{txfonts}
\usepackage{booktabs}
\usepackage[T1]{fontenc}
\usepackage{babel}
\usepackage{enumerate}
\usepackage{graphicx, subfigure}
\usepackage{float}
\usepackage{subeqnarray}
\usepackage{datetime}
\usepackage{eurosym}
\usepackage{mathrsfs}
\usepackage{wasysym}

\graphicspath{ {figures_cal/}{figures/}{figures_ml/}{/}{figures_new/}}

\newtheorem{theorem}{Theorem}
\newtheorem{proposition}{Proposition}
\newtheorem{lem}{Lemma}
\newtheorem{cor}{Corollary}

\def\calI{{\cal I}}
\def\calT{{\cal T}}

\def\bone{\mathbf{1}}

\def\Ind{{\calI}}

\def\be{\begin{eqnarray}}
\def\ee{\end{eqnarray}}

\def\IdioComp{{\cal W}}

\def\Esp{{\mathbb{E}}}
\def\EM#1{{\left<#1\right>}}

\DeclareMathOperator*{\argmin}{arg\,min}

\author{Emmanuel Bacry\footnote{CMAP CNRS-UMR 7641 and Ecole Polytechnique, 91128 Palaiseau}, Adrian Iuga\footnote{CREST and LAMA CNRS-UMR 8050, 92245 Malakoff}, Matthieu Lasnier\footnote{Kepler-Cheuvreux}, Charles-Albert Lehalle\footnote{Capital Fund Management, Paris and Imperial College, London}}
\date{Printed \today}

\begin{document}

\title{Market impacts and the life cycle of investors orders}
\maketitle

\begin{abstract}
  In this paper, we use a database of around 400,000 metaorders issued by investors and electronically traded on European markets in 2010 in order to study market impact at different scales.

  At the intraday scale we confirm a square root temporary impact in the daily participation, and we shed light on a duration factor in $1/T^{\gamma}$ with $\gamma \simeq 0.25$. Including this factor in the fits reinforces the square root shape of impact.
We observe a power-law for the transient impact with an exponent between $0.5$ (for long metaorders) and $0.8$ (for shorter ones). Moreover we show that the market does not anticipate the size of the meta-orders.
The intraday decay seems to exhibit two regimes (though hard to identify precisely): a ``slow'' regime right after the execution of the meta-order followed by a faster one.
At the daily time scale, we show price moves after a metaorder can be split between realizations of expected returns that have triggered the investing decision and an idiosynchratic impact that slowly decays to zero. 

  Moreover we propose a class of toy models based on Hawkes processes (the Hawkes Impact Models, HIM) to illustrate our reasoning. 
  We show how the Impulsive-HIM model, despite its simplicity, embeds appealing features like transience and decay of impact. The latter is parametrized by a parameter $C$ having a macroscopic interpretation: the ratio of contrarian reaction (i.e. impact decay) and of the "herding" reaction (i.e. impact amplification).
\end{abstract}



\section{Introduction}
\label{sec:definitions}

Interactions between the Price Formation Process (PFP) and market microstructure are of paramount importance: 
since regulatory evolution (Reg NMS in the US and MiFID 1 and now MiFID 2 in Europe, promoting competition among exchanges via an extensive use of electronic orderbooks) and the 2008 liquidity crisis, academics and market participants try to understand how to align the orderbook dynamics with the interests of final investors.

The market impact of large {\em metaorders}\footnote{A {\em metaorder} refers to a large investor buy or sell order which is executed as a succession of smaller orders.} is at the heart of market-microstructure regulation discussions.
Thus, quantifying this market impact has become a major issue. 
On the one hand, it seems reasonable to assume information processed by investors, driving their decisions, has a predictive content on future prices. Hence the price should move during the execution of their metaorders, in a detrimental direction (i.e. up if they buy, or down if they sell).
On the other hand, large orders lead to liquidity imbalance, mechanically moving the price in the same direction.
Part of this mechanical move generates trading costs. The important point is that trading costs prevent investors to initiate some trades, or at least decreases the performance of their trades; 
thus, better understanding of market impact should lead to  changes of  market microstructure regulations 
in order to lower the trading costs and consequently  improve the allocation of investors in listed firms.

The gap between informational move and mechanical reaction to trading pressure can be empirically investigated. However, available investors' metaorder databases do not include long track records.  Metaorders have been recorded in a systematic way only recently. 
In order to give an empirical answer to the nature of market impact,
accurate market data are not enough, a clear identification and time stamping of metaorders is needed too.
Few studies had access to this kind of database: \cite{citeulike:4325901}, \cite{citeulike:12838207}, \cite{doi:10.1080/14697680903373692}, \cite{citeulike:9771410}, \cite{TLDLKB}, \cite{bouchaud09}, \cite{FarmerLillo2006PriceImpact}, \cite{citeulike:13266538} and  \cite{lastlillo}. 

In this paper we report measures made on a database coming from the brokerage arm of a large European Investment Bank, whose trading flow was around 5\% of European investors' flow on equity markets at the time the metaorders have been recorded.
As it is detailed in Section \ref{sec:database}, this database is made of very coherent metaorders in terms of trading style, trading universe and market context.
This paper provides information on the impact of large orders on the PFP at many scales:
\begin{itemize}
\item at the scale of each metaorder: the \emph{temporary impact} and its relationship with explanatory variables are documented;
\item at a lower time scale: the price reaction to trading pressure (i.e. the \emph{transient impact}) is measured, 
  and the relaxation of prices once metaorders end (i.e. the \emph{impact decay}) is also studied.
\item at a larger scale: we zoom out days after the metaorder ends to document the \emph{permanent impact} in an attempt to disentangle it from the informational price move.
\end{itemize}
Besides, this paper provides a \emph{toy model} to discuss stylized facts discovered in the database. It is based on Hawkes processes, since they have demonstrated in previous papers to be well suited for orderbook dynamics modelling (\cite{citeulike:9217792}, \cite{Hewlett1}, \cite{BM}, \cite{HBB}).

Our paper is organized as follows.
Section 2 introduces metaorder and market impact related definitions. It describes the main database and the various filters that will be applied to it in order to obtain subset databases used all along the papers. It finally present the main market impact estimation principles and a first estimation of the whole market impact curve. Sections 3, 4, and 5  presents the results we obtained respectively on the temporary, the transient and the decay market impact. In Section 6, we introduce a Hawkes Impact Model (HIM) which is able to reproduce most of the stylized facts previously observed on both transient and decay market impacts. Permanent impact results and discussions can be found in Section 7 and we conclude in Section 8.


\section{Definitions, Database and Market impact estimation principles}
\label{sec:defs}

%

%
\subsection{Basic definitions}
\label{sec:defimp}
Let $\Omega$ be the set of all the available meta-orders. Let us point out that we will only deal with metaorders that are fully executed within a single day. Let $\omega \in\Omega$ such a metaorder executed on an instrument $S$ and on a given day $d$. We consider that its execution starts at (intraday) time 
$t_0$ and ends the same day at time $t_0 + T$.
The metaorder $\omega$ has a size $v$ which corresponds to the number of shares effectively executed. We note 
$v_t$ the number of shares that have been executed up to time  $t\in [t_0,t_0+T]$.
The sign of $\omega$ will be noted $\epsilon$ (i.e. $\epsilon=1$ for a buy order and -1 for a sell order).
Let us point out that all the quantities introduced in this section depends implicitely on $\omega$ (we should write $t_0(\omega)$, $T(\omega)$, \ldots). For the sake of simplicity, we chose to ommit this dependence whenever there is no ambiguity. 

Thanks to external market data, we also have access to the daily volume $V$  traded the same day $d$ on the instrument $S$ (summed over all the European trading venues), and to the volume traded $\bar v$, by the whole market (including the metaorder) between time $t_0$ and $t_0+T$.  We use ${\bar v}_t$ for the market volume traded between $t_0$ and an arbitrary $t>t_0$. The daily volatility (estimated using the Garman-Klass approach \cite{GAR80} on intraday 5 minutes bars) is $\sigma$.
The average bid-ask spread (on the aggregated venues) the same day is $\psi$.

These quantities can be combined to obtain:
\begin{itemize}
\item \emph{the trading rate} $r=v/{\bar v}$;
\item and \emph{the daily participation} $R=v/V$.
\end{itemize}
In order to be less dependent on the market activity during the execution of the metaorder, we will sometimes use the \emph{average participation rate} $\dot\nu=R/T$ as a replacement for $r$.

\subsection{Market impact definitions}
\label{sec:defimp}
The \textit{market impact curve} of a meta-order $\omega$ quantifies the magnitude of the (relative) price variation which is {\bf due to the meta-order} $\omega$ between the starting time of the meta-order $t_0$ and the current time $t>t_0$.
Like all other authors, we assume that the impact of the execution of $\omega$ on the price is additive.
Let $\Delta P_{t}(\omega)$ be a proxy for the realized price variation between time $t_0$  and time $t_0+t$, we thus write
\begin{equation}
  \label{eq:cam:MIdef}
  \epsilon(\omega) \Delta P_t(\omega)= \eta_t(\omega) + \epsilon(\omega) \Delta W_t(\omega),
\end{equation}
where $\eta_t(\omega)$ represents the market impact curve and $\Delta W_t$ the exogeneous variation of the price : it corresponds to the relative price move that would have occurred if the metaorder was sent to the market (and consequently not executed).

Let us note that the choice for the proxy $\Delta P_{t}$ does not influence the results we obtained all along the paper. We decided to follow Almgren \cite{citeulike:4325901} and Bershova $\&$ Rakhlin  \cite{citeulike:12838207} and use the {\em return proxy} defined by
$$
\Delta P_{t} = \frac {P_t-P_{t_0}}{P_{t_0}},
$$
where $P_t$ represents the mid-price of the instrument at time $t$.
Let us point out that, in order to obtain more heterogeneous prices, some authors (Lillo and Farmer \cite{citeulike:9771410}) prefer to combine the log-return and the spread relative proxies:
they choose to divide the log-return proxy  by $\psi(\omega)$.

\vskip .3cm
\noindent
Finally, we introduce the following terminology in order to address the different parts of the market impact curve :
\begin{itemize}
\item $\eta_{t_0+T}(\omega)$ corresponds to the \textit{temporary market impact}, i.e., the impact at the end of the meta-order execution,
\item $\{\eta_t(\omega)\}_{t_0\le t \le t_0+T}$ corresponds to the \textit{transient market impact curve}, the impact curve during the execution of the meta-order,
\item $\{\eta_t(\omega)\}_{t_0+T\le t}$ corresponds to the \textit{decay market impact curve}, the impact curve after the execution of the meta-order and
\item $\eta_{t}(\omega)$, where $t \gg t_0+T$ corresponds to the \textit{permanent impact}. What we exactly mean by the limit $t \gg t_0+T$ will be make clearer in Section \ref{sec:permanent}.
For now, it is sufficient for the reader to know that it refers to an extraday time limit, sufficiently far from the ending-time $t_0+T$ of $\omega$, so that the impact of $\omega$ can be considered as no more influenced by trading mechanics.
\end{itemize}

\subsection{The main metaorder database}
\label{sec:database}

The database $\Omega$ is made of nearly 400,000 metaorders.
The selected metaorders have been traded electronically by a single large broker during year 2010 on European markets.
We built different databases from the original 400,000 database (see Table \ref{tab:filters} for complete list of all database and associated filters)  in order to adapt to the part of the market impact curve under study and to have as much orders as possible for each time scale :
\begin{itemize}
\item For intraday studies, we only kept orders traded by trading algorithms whose trading rate is as much as possible independent from the market conditions. Note in \cite{lastlillo}, authors underline the potential influence of the trading rate on the market impact components. It allows to avoid sudden accelerated trading rates having an hidden influence on the price moves. Hence we kept VWAP, TWAP, PoV and only few Implementation Shortfall instances (see Appendix \ref{sec:algonames} for definitions of these acronyms).
\item We kept orders large enough to protect our estimation from noise.
\item For the decay market impact curve, we needed the metaorder execution to stop halfway to the market close in order to observe prices relaxation long enough between the end of the metaorder and the closing time.
\end{itemize}
We ended up with five main databases: $\Omega^{(te)}$ to study temporary impact, $\Omega^{(tr)}$ to study transient impact and impact decay, and $\Omega^{(de)}$ to study decay impact and $\Omega^{(day)}$ to study daily effects. 
The filter \emph{consistency with market data} means that we reject metaorders such that we cannot identify their child orders in the market data within the same second (i.e. when a child order is sent or is partially fill or fill, we try to find the correposning modification in the market data: same price, same quantity, same second; if we cannot, we reject the metaorder).
See Appendix \ref{sec:datatables} for detailed statistics on these databases.

\begin{table}[!h]
\centering
	\begin{tabular}{|c|c|c|}\hline
Databases $\Omega$ & Filters & Number of meta-orders\\\hline
Original database $\Omega$ & $\emptyset$ & \textbf{398.812} \\\hline
Daily database $\Omega^{(day)}$ & smooth execution condition & \textbf{299.824} \\\hline
\multirow{2}{*}{Database for temporary impact $\Omega^{(te)}$} & 10 atomic orders minimum & 191.324\\
& consistency with HF market data & \textbf{157.061}\\\hline
\multirow{4}{*}{Database for transient impact $\Omega^{(tr)}$} & number of daily trades on the stock $\geq 500$ & 150.100\\
& $T > 3$ minutes & 134.529\\
& $r\in[3\%,40\%]$ & 94.818\\
& $R\in[0.1\%,20\%]$ & \textbf{92.100}\\\hline
Database for decay impact $\Omega^{(de)}$ & $t_0(\omega)+2T(\omega)<$ closing time & \textbf{61.671}\\\hline
\end{tabular}
\caption{Different databases and filters used to obtain them. See Tables \ref{tab:data:dist} and \ref{tab:stat:all} in Appendix \ref{sec:datatables} for more statistics on those databases.}
  \label{tab:filters}
\end{table}

\subsection{Market impact estimation principles and first estimation}
\label{sec:principles}
The market impact curve $\eta_t$ as defined in Section \ref{sec:defimp}
is clearly not directly observable.
As already explained, the exogeneous price move $\Delta W_t$ in Eq. \eqref{eq:cam:MIdef} corresponds to the price move that would have occurred if the metaorder was not placed at all.
If $\Delta W_t$ was entirely independent of the decision process which lead the agent to place the metaorder $\omega$, then it would correspond to an independent noise term that could be reduced by averaging on a large number of metaorders. However, 
the investor shares information at the root of his decision process with other investors, and consequently $\Delta W_t$ is often not independent of the \emph{decision} that generated $\omega$. Indeed, consider as an example the case of a trend-following agent: he will place a buying meta-order only if he forecasts an increase in the price, and his forecast is generally not independant of the realized price move $\Delta W_t$.

From a market impact estimation point of view, this term corresponds to a non-independent additive ``noise term''.
This ``informational pollution'' of market impact measurement is eluded in some papers (like \cite{citeulike:4325901}), and discussed in others (like \cite{citeulike:9771410}, \cite{citeulike:12825932}, \cite{citeulike:13266538}).
In our paper, we shall neglect it when ``only'' studying the relative dependence of the impact curve on some explanatory variables, i.e.,
for intraday measurements (see Sections \ref{sec:temporary}, \ref{sec:transient} and \ref{sec:decay}), but we shall take it into account when estimating ``absolute'' levels of impact, i.e., 
for daily estimates (Section \ref{sec:permanent}).

In Fig. \ref{fig:prox}, we show a first estimation of the market impact curve using the database $\Omega^{(de)}$. This estimation corresponds to the quantity
\begin{equation}
  \label{eq:cal:avg1}
  {\hat \eta_s}=\EM{\epsilon(\omega)\Delta P_{t}(\omega)}_{T=1},
\end{equation}
where $\EM{\ldots}_{T}$ stands for a three-step procedure : 
\begin{itemize}
\item Compute  the quantity in between bracket for each $\omega$
\item For each $\omega$, perform a time-rescaling on this quantity such that the duration $T(\omega)$ is rescaled to the same fixed value $T$ \item Average on all these time-rescaled quantities. The so-obtained rescaled time is referred to as $s$ whereas $t$ denotes the physical time. Thus $s=T$ (in Fig. \ref{fig:prox} we chose $T=1$ so $s=1$)
refers to the time the executions of the meta-orders are over.
\end{itemize}
As already explained, this estimation relies on the assumption that the ``exogenous market moves'' $\Delta W_t$ cancel once averaged. Meaning that, as a random variable, $\Delta W_t$ should have finite variance and basically satisfy 
$\Esp(\epsilon(\omega)\Delta W_t(\omega))=0$ (and not $\Esp(\Delta W_t(\omega))=0$).

Let us point out that, on Fig. \ref{fig:prox}, one can easily identify the two sections of the market impact curve : the concave transient part (during execution) and the convex decay part (after execution).
\begin{figure}[!h]
  \centering
  \includegraphics[width=.8\linewidth]{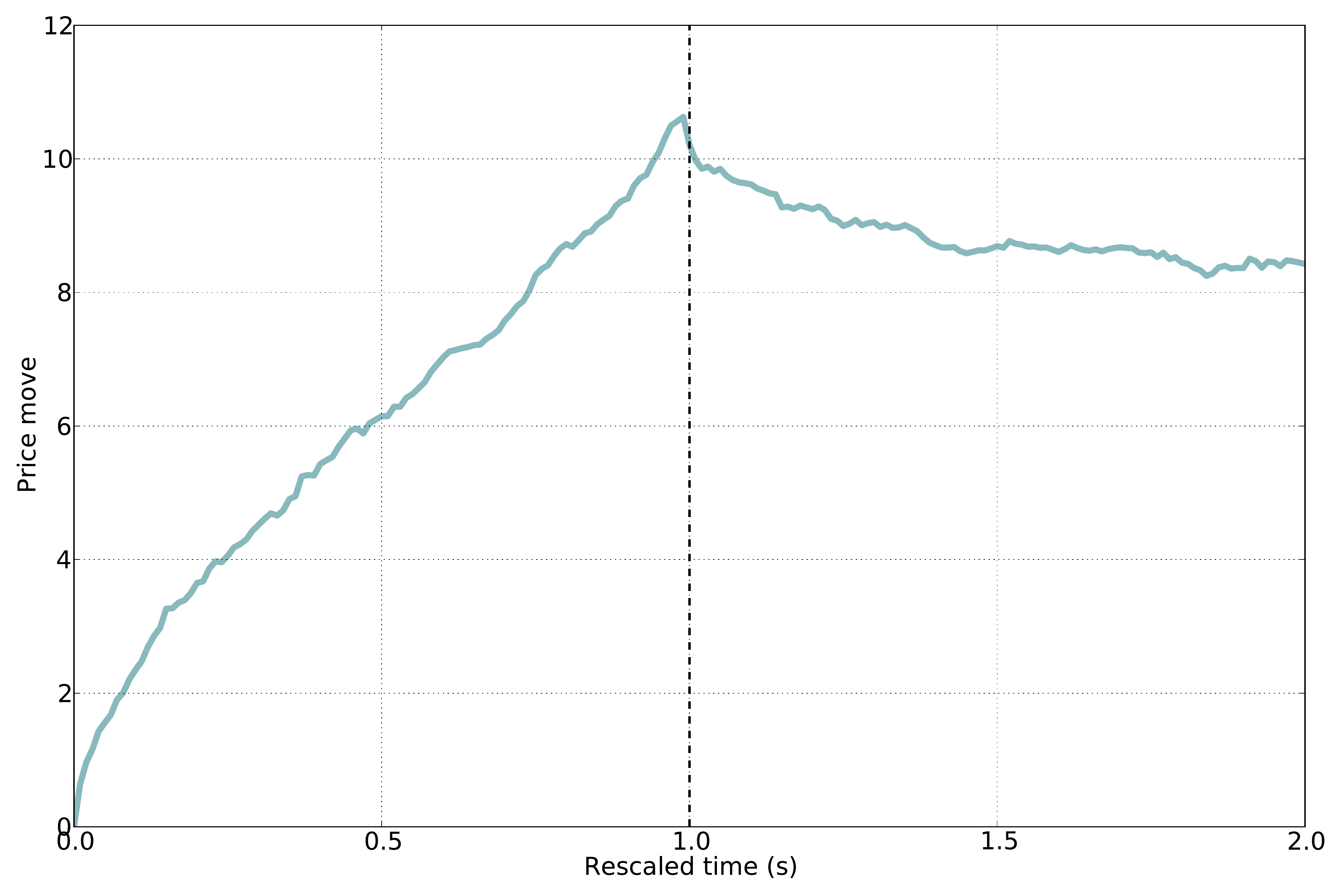}
  \caption{The average market impact curve \eqref{eq:cal:avg1} for $T=1$: on the horizontal axis we use the rescaled time $s$ (so that $s=1$ corresponds to the  time the execution of the meta-order is over).}
    \label{fig:prox}
\end{figure}

\paragraph{Studying the influence of an explanatory variable through averaging.}
One could follow the same lines to study the influence on the market impact curve of a given explanatory variable $X(\omega)$.
For that purpose, one would use an ``averaging'' of equation \eqref{eq:cam:MIdef} over ``slices'' of the variable $X$.
For instance one can split the possible values of the daily participation $X(\omega) = R(\omega)$ in intervals $q^X_1,\ldots,q^X_K$ naturally associated to $K$ quantiles and compute,
as an estimate of the temporary market impact, 
\begin{equation}
  \label{eq:cal:avg}
  {\hat \eta}_X(k)=\EM{\epsilon(\omega)\Delta P_T(\omega)~|~X(\omega)\in q^X_k}_{T}
\end{equation}
where $\EM{\ldots}_{T}$ corresponds to the notation used in Eq. \eqref{eq:cal:avg1}.
Of course, as explained before,
this approach relies on the assumption that the ``exogenous market moves'' $\Delta W_t$ cancel once averaged, i.e., essentially 
$\Esp(\epsilon(\omega)\Delta W_t(\omega)~|~X(\omega)\in q)=0$.


This averaging procedure allows to identify a dependence to a single explanatory variable only. To be able to regress the market impact on more than one explanatory variable, for instance three variables, one should use $K$ quantiles for each variable, i.e. averaging on $K^3$ subsets (i.e. hypercubes) of $\mathbb{R}^3$. The number of metaorders used to estimate one $ {\hat \eta}(k_1,k_2,k_3)$ would then be divided by $K^3$. Thus the variance of the estimator would be roughly  multiplied by $K^3$! So this procedure is not adequate for studying the dependency on more than one explanatory variable. The ``Direct regression'' procedure described below is more adequate.

Let us point out that, if linear regressions are performed after the previously described averaging procedure, the usual statistical measures of regression quality will not have the usual meaning. The ``$R^2$'' for instance can be very high, even if the relation between the impact and the explanatory variables is very noisy (since the averaging step leverage on the central limit theorem to cancel most of the ``noise''). 

This averaging methodology is nevertheless useful to draw figures and obtain qualitative results. 
It will be used through this paper.

\paragraph{Direct regression.}
An alternative approach is to fit directly an explicit model on Eq. \eqref{eq:cam:MIdef}. Let's say we want to fit the parameters of a power law linking the market impact to the explanatory variable $X=R$, i.e., the daily participation of the metaorder.
It leads to the following parametric version of \eqref{eq:cam:MIdef}:
\begin{equation}
  \label{eq:cal:directest}
  \epsilon(\omega)\Delta P_T(\omega)=a\cdot X(\omega)^\gamma + \epsilon(\omega)\Delta W_T(\omega).
\end{equation}
As far as we assume once more that the exogenous prices moves are independent enough from $\omega$ to be averaged, we can select a distance $d(\cdot,\cdot)$ in the space of returns $\Delta P$ and use any minimization method to obtain ``fitted'' parameters $(a^*,\gamma^*)$:
\begin{equation}
  \label{eq:fitAG}
  (a^*,\gamma^*) = \arg\min_{(a,\gamma)} \EM{d\left( \epsilon(\omega)\Delta P_T(\omega), a\cdot X(\omega)^\gamma \right)^2}
\end{equation}
where $\EM{\ldots}$ corresponds to the empirical mean of the corresponding quantity computed on the set of all the metaorders. This approach will be also used in the case we want to study the dependency on more that one explanotory variable.

This approach uses all the points in the database (and not their averaged version), and thus produces more accurate results. 
Of course, as for the averaging method introduced above, it basically relies on the 
independence between metaorder initiation (starting time) and the exogeneous market moves ($\Delta W_T$).
Under this hypothesis, let us point out that the regression residuals has a very simple interpretation : it corresponds to an estimation of $\Delta W_T$.
It allows to test, a posteriori,  the independence hypothesis.

In the following, we shall use the usual $L^2$ distance or the $L^1$ distance (which is better suited for dealing with fat tails and/or rare but intense events).


\section{The temporary market impact}
\label{sec:temporary}

\subsection{Selection of the explanatory variables}
The temporary market impact has been mainly studied from three viewpoints:
\begin{itemize}
\item As the main source of trading costs. The obtained model can then be used in an optimal trading scheme (see \cite{citeulike:4325901}, \cite{doi:10.1080/14697680903373692} and \cite{citeulike:6728746} 
and for link with optimal trading see  \cite{OPTEXECAC00}, \cite{citeulike:10363463} and \cite{citeulike:5797837}), or used by an investment firm to understand its trading costs (like in \cite{citeulike:4368376}, \cite{citeulike:12838207}, \cite{citeulike:13266538} or \cite{citeulike:12802757} written by author involved in investment firms).
\item It can be viewed as an important explanatory variable of price discovery and studied as such, often by economists, like the seminal work of Kyle \cite{citeulike:3320208} or later in \cite{citeulike:10138505} or \cite{citeulike:98140}.
\item Last but not least, statistical tools have been built to be able to estimate the temporary market impact at the scale of one trade (see \cite{BM} or \cite{citeulike:9574751}). The implicit conditioning of such ``atomic'' orders by metaorders is sometimes discussed in such papers, but it is not their main goal.
\end{itemize}

One common point of these studies is that the temporary market impact of a metaorder $\omega$ of size $v(\omega)$  includes three main components :
\begin{itemize}
\item A component reflecting the \emph{size of the metaorder}, resized by something reflecting the volume in the orderbooks of the traded security.  The daily participation $R(\omega)$ or the trading rate $r(\omega)$ should capture most of the dynamics of this component. 
\item A component rendering the \emph{uncertainty on the value of the underlying price} during the metaorder. The volatility
during the metaorder $\sigma(\omega)$ or the bid-ask spread $\psi(\omega)$ are typical measures for this.
\item And a last component  that captures the \emph{information leakage} generated by the metaorder, a good proxy being its duration $T(\omega)$.
\end{itemize}
Let us point out that all authors found multiplicative relations between each of these components and their corresponding explanatory variable, so, in the scope of an estimation using the averaging methodology, we expect a linear dependence of the temporary market impact $\eta_{T}(\omega)$ on the logarithm of these explanatory variables.

As an example, the left plot of Figure \ref{fig:scatter} displays 
$\hat \eta_{X=R}$ as defined in Eq. \eqref{eq:cal:avg} versus the daily participation
$R$ when averaging on all the metaorders of $\Omega^{(te)}$.
It clearly shows the influence of $R$ on the impact : the higher the daily participation rate the higher the impact. The same result is found when replacing the daily participation $R$ by the trading rate $r$ (see right plot of the same figure).

\begin{figure}[!h]
  \centering \includegraphics[width=1\linewidth]{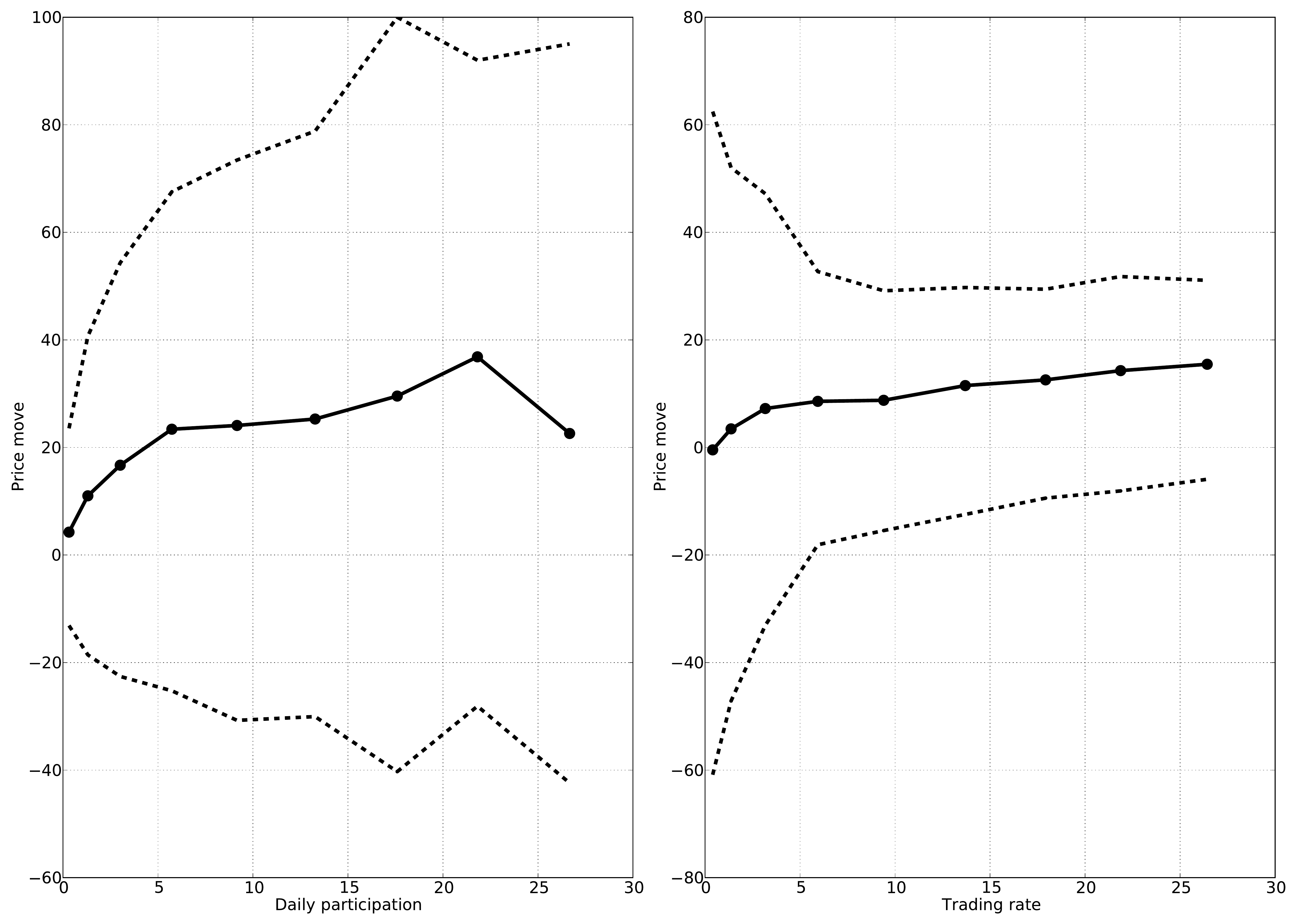}
  \caption{The estimated market impact $\hat\eta_X$ (as defined in Eq. \eqref{eq:cal:avg}) as a function of the daily participation $X=R$ (left) or of the trading rate $X=r$ (right). Each point is the average of one decile of the X variable, dotted lines are $25\%$ and  $75\%$ quantiles, showing the amplitude of market moves.} 
  \label{fig:scatter}
\end{figure}



\subsection{Numerical results}
To study temporary market impact and its dependence to  explanatory variables, we use the $\Omega^{(te)}$ database (see Table \ref{tab:filters}) and followed the direct regression approach described in Section \ref{sec:principles}.
As an example, we tested the dependence in the daily participation $X=R$, since it has been identified as significant by other papers. It means that we fitted Eq. (\ref{eq:cal:directest}):
$$\epsilon(\omega)\Delta P_T(\omega)=a\cdot R(\omega)^\gamma + \epsilon(\omega)\Delta W_T(\omega),$$
and found an exponent $\gamma\simeq 0.449$ using the $L^2$ distance and a lower exponent (around $0.40$) using the $L^1$ distance (see regression {\bf (R.1)} in Table \ref{tab:res:r}). The fact that we obtain different estimation when using the two distances $L^1$ and $L^2$
can be explained by the fact that the joint distribution of $\epsilon \Delta P$ and $R$ is skewed to large values of $\Delta P$. The $L^2$ distance has no other choice than to render this skewness by setting an high value to $\gamma$, while the $L^1$ distance focuses more on the center of the distribution. The source of this skewness  could stem from an informational effect as a dependence between $\epsilon$ and $W_T$. However, we do not have enough elements to conclude on this point.

Once the power exponent corresponding to a given explanatory variable $X$ has been estimated as above, the effect of any other explanatory variable $Y$ can be explored and shown by averaging the residuals over quantiles of $Y$ itself. We will refer to the associated chart as the \emph{trace of $Y$ on the residuals}.

Top (resp. bottom) plot of Fig. \ref{fig:trace:duration} shows the trace of the duration $Y=T$ of the metaorders on the residuals of the $X=R$ daily participation regression using $L^2$ (resp. $L^1$) distance.
The slope is clearly negative.
Moreover, one can notice the bump around 0.6 days (corresponding to 6 hours after opening). We suspect it stems from the macroeconomic news one hour before the opening of NY markets, 6 hours after the opening of European ones. At the open of US markets, volatility is higher and the dependence between the side of the metaorder and the market move can be higher too.

\begin{figure}[!h]
  \centering
  \includegraphics[width=1\linewidth]{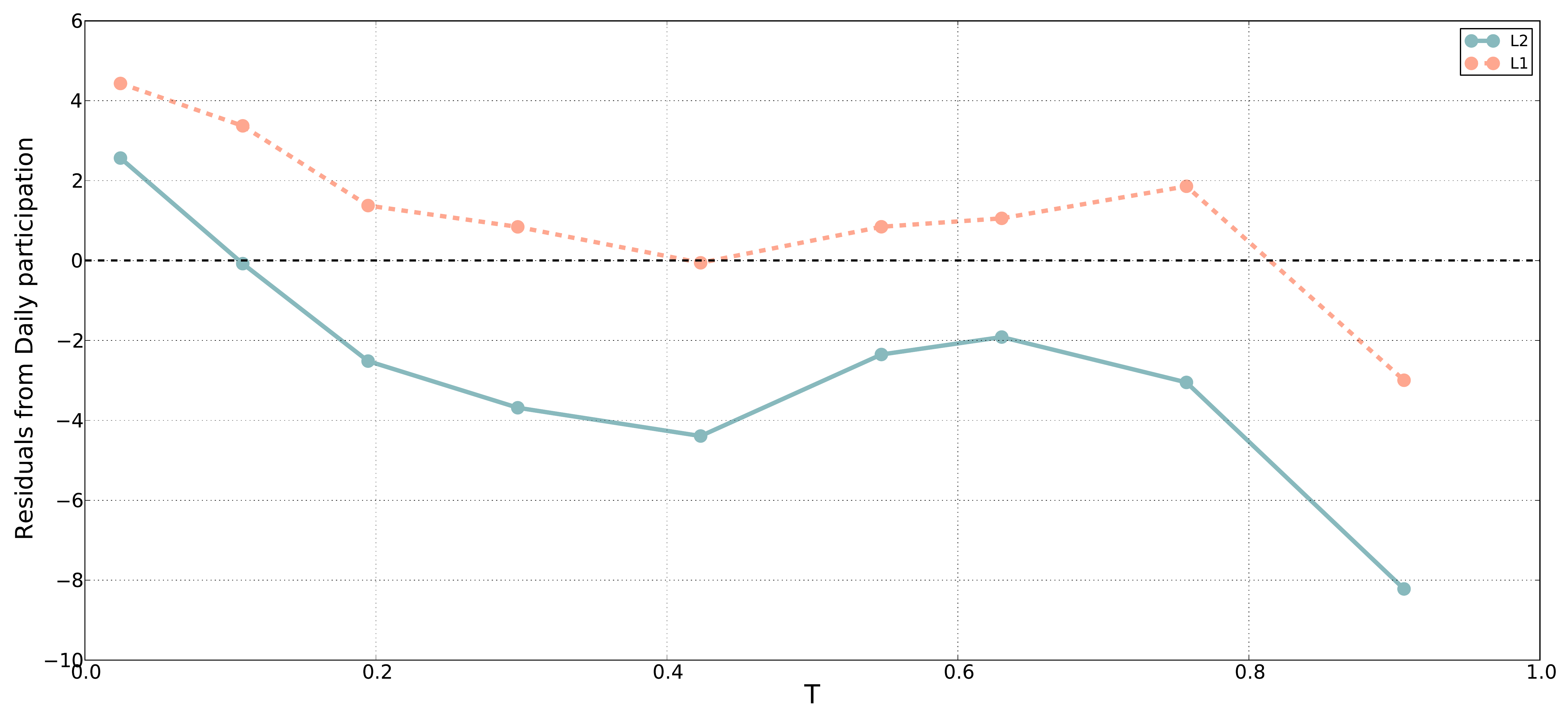}
  \caption{Trace of the duration $Y=T$ on the residuals of the $X=R$ daily participation regression. Top: using an $L^2$ metric, bottom: using a $L^1$ metric.} 
  \label{fig:trace:duration}
\end{figure}

\begin{table}[h!]
  $$
\begin{array}{|crrrr|}\hline
\mbox{Regression} & \mbox{Parameter} & \mbox{Coef. (log-log)} & \mbox{Coef. (L2)} & \mbox{Coef. (L1)} \\\hline
{\bf (R.1)} & & & & \\
& \mbox{Daily participation} &  0.54 &  0.45 &  0.40 \\
\hline
{\bf (R.2)} & & & & \\
& \mbox{Daily participation} &  0.59 &  0.54 &  0.59 \\
& \mbox{Duration} & -0.23 & -0.35 & -0.23 \\
\hline
{\bf (R.3)} & & & & \\
& \mbox{Daily participation} &  0.44 &  - & - \\
& \mbox{Bid-ask spread} &  0.28 &  - & - \\
\hline
{\bf (R.4)} & & & & \\
& \mbox{Daily participation} &  0.53 &  - &  - \\
& \mbox{Volatility} &  0.96 & - &  - \\
\hline \hline
{\bf (R'.1)} & & & & \\ 
& \mbox{Trading rate} &  0.43 &  0.33 &  0.43 \\
\hline
{\bf (R'.2)} & & & & \\
& \mbox{Trading rate} &  0.37 &  0.56 &  0.45 \\
& \mbox{Duration} &  0.15 &  0.24 &  0.23 \\
\hline
{\bf (R'.3)} & & & & \\
& \mbox{Trading rate} &  0.32 &  - & - \\
& \mbox{Bid-ask spread} &  0.57 &   -& - \\
\hline
{\bf (R'.4)} & & & & \\
& \mbox{Trading rate} &  0.32 &  - & - \\
& \mbox{Volatility} &  0.88 &  - &  - \\
\hline
\end{array}$$

  \caption{Direct regression approach algorithm described in Section \ref{sec:principles} of the temporary market impact for various sets of explanatory variables. For each set, the power exponent estimations are given using $L^1$ distance, $L^2$ distance and regular log-log regressions. There is an horizontal line $-$ when there is not significant difference between the three regressions.
}
  \label{tab:res:r}
\end{table}

To confirm the dependence in $T$, we fit a power law on the trading rate $X=r$ instead of the daily participation $R$:
$$\Delta P_T(\omega)=a\cdot r(\omega)^\gamma + \epsilon(\omega)\Delta W_T.$$
We found respectively power 0.43, 0.33 and 0.42 for log-log regression, $L^2$ minimization and $L^1$minimization
(see regression {\bf (R'.1)} in Table \ref{tab:res:r}). The trace of $Y=T$ on the residuals of this regression (see Figure \ref{fig:trace:duration:nu}) exhibits a positive slope, confirming the way duration $T$ can be used in combination to daily participation or trading rate to improve the modelling of impact.

As a comparison with other academic studies:
\begin{itemize}
\item \cite{citeulike:9771410} use the bid-ask spread $\psi$ as prefactor to the daily participation (i.e. regression {\bf (R.3)} for us) they find (for the daily participation rate) a power of 0.44 to 0.48 for metaorders executed on the Spanish stock exchange and 0.64 to 0.72 for metaorders executed on the London stock exchange. We find it fits better with $\psi^{0.3}$ instead of $\psi$ and identify a power of $0.4$ for the daily participation.
\item \cite{citeulike:12825932} arbitrary assume a prefactor which is the volatility $\sigma$ and a square root law for the daily participation rate. Our regression {\bf (R.4)} of Table \ref{tab:res:r} agrees with a linear effect in volatility and a power around 0.45 for $R$
(the square root law would correspond to an exponent of 0.5). \cite{citeulike:12838207} use a model linear in volatility and a power of $0.47$. With the same volatility prefactor, \cite{ATHL} find a (larger) power of $0.6$.
\end{itemize}

\begin{figure}[!h]
  \centering
  \includegraphics[width=1\linewidth]{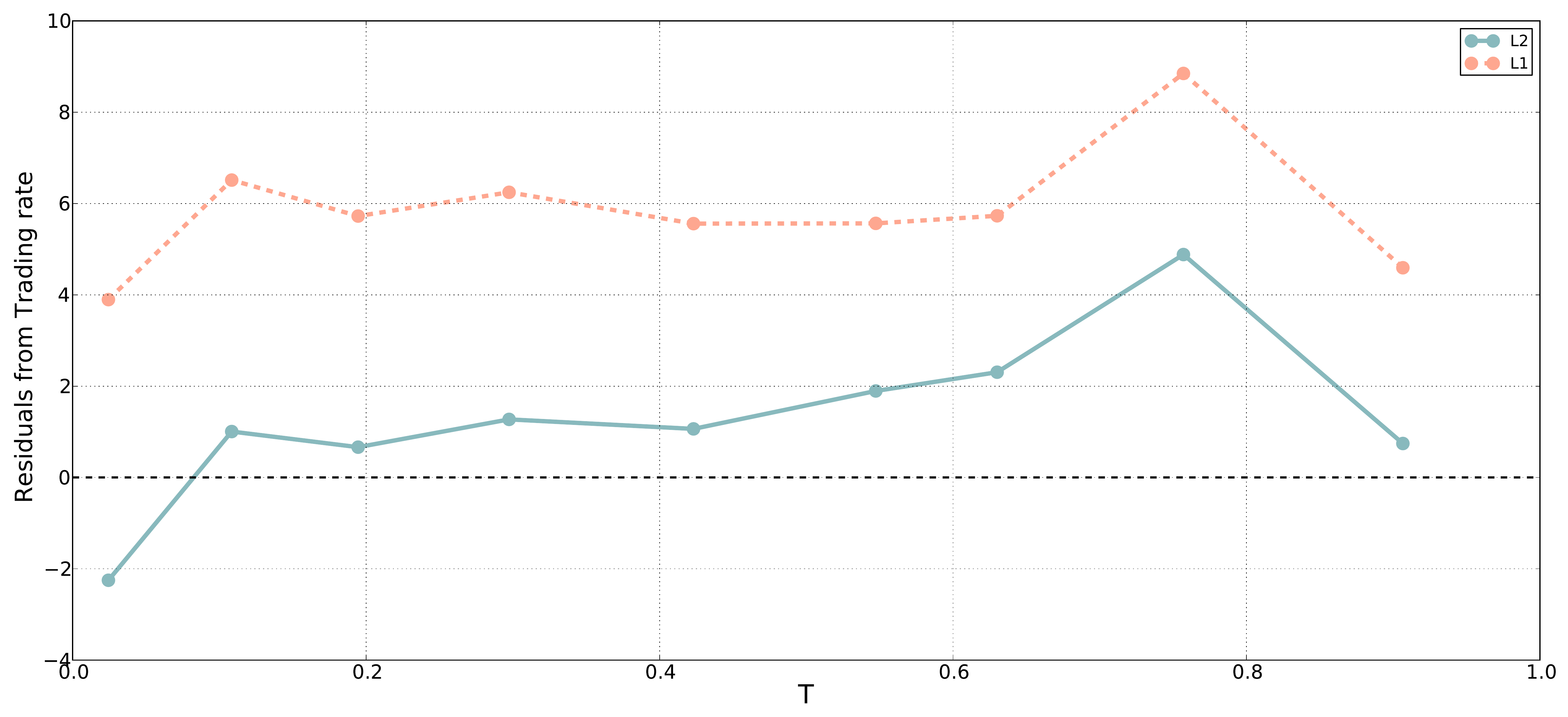}
  \caption{Trace of the duration $Y=T$ on the residuals of the trading rate $X=r$ regression. Top: using an $L^2$ metric, bottom: using a $L^1$ metric.} 
  \label{fig:trace:duration:nu}
\end{figure}

We performed direct fit of other explanatory variables and found the usual dependencies in $\sigma$ and $\psi$.
Table \ref{tab:res:r} gives the results of direct regression approach algorithm described in Section \ref{sec:principles} for various sets of explanatory variables. For each set, the power exponent estimations are given using $L^1$ distance, $L^2$ distance and regular log-log regressions. 

Though 
our investigations seem to confirm the celebrated \emph{square root law} according to which the temporary market impact is a function of the square root of the daily participation rate (see {\bf (R.1), (R.2), (R.3),
(R.4)} in Table \ref{tab:res:r}),
it nevertheless shows that the exponent vary with the metric used for the fit (being lower for a direct $L^2$ estimation instead of a log-log fit and even lower for a $L^1$ distance), underlying potential asymmetric regimes in the impact formation.
However, our result depart from the \emph{square root law} in the sense that we
clearly identified a duration effect via a $1/T^{\gamma}$ factor with $\gamma \simeq 0.25$ (see {\bf (R.2)}).
This influence of the duration of the metaorder may be difficult to measure when the trading process is too opportunistic (especially when the duration is a function of apparent prices opportunities during the trading), explaining why all authors did not noticed it. Our database being dominated by trading algorithms with more stable benchmarks (VWAP, PoV), we have been able to capture this mechanical effect of trading pressure on price formation.

The regressions using the daily participation $R$ and the ones using the trading rate $r\simeq R/T$ are consistent (especially the fit obtained via direct regressions rather than for ones using log-log ones, that are biased, as explained earlier); for instance:
{\bf (R.2)} states $\eta\simeq \sqrt{R}/T^{0.25}$ and {\bf (R'.2)} reads $\eta\simeq \sqrt{r} \times T^{0.25}$ (for direct regressions). 

Adding volatility or bid-ask spread to an explanation of temporary market impact by the participation rate $r$ (see regressions {\bf (R'.3)} and {\bf (R'.4)}) does not change its exponent: $\eta\simeq \psi^{1/2} r^{1/3} \simeq \sigma r^{1/3} $, but provoke a small change in the exponent of the trading rate $R$: $\eta\simeq \psi^{0.3} R^{0.44} \simeq \sigma\sqrt{R} $. It is not sure this difference is generic. Attempt to use more than two variables in regressions did not provided significant relationships.

\section{The transient market impact curve}
\label{sec:transient}

\subsection{A power law fit}
\label{sec:tmistat}
Previous section confirms, as many other empirical studies before us, that the temporary market impact of a metaorder of size $v$ is proportional to ${v}^\gamma$ (with $\gamma$ close to 0.5). 
Thus, we expect the temporary market impact of the first half of the execution (the first $v/2$ contracts) to be more important than the one of the second half (the last $v/2$ contracts). 
Generalizing this argument to any portion of the metaorder, we expect the transient market impact curve to be a concave function of the time. 
The first empirical study confirming this intuition is due to Moro et \textit{al.} (\cite{citeulike:9771410}), lately, it has also been  confirmed  by the work of Bershova $\&$ Rakhlin (\cite{citeulike:12838207}). In both cases, behavior close to power-laws were found.
Let us point out that the   \emph{latent order book model} of \cite{citeulike:10477292} can be seen as a possible qualitative explanation of this well established stylized fact. In this model the agents place limit orders only when the price is close enough to their vision of the price. Thus, more and more liquidity is revealed as the price is trending, it results in  "slowing down" this trend.

In this section we use our $\Omega^{(tr)}$  database (see Table \ref{sec:database}) and first confirm the concavity of the transient market impact curve.
Apart from this well known stylized fact, we study the link between the curvature of the transient market impact curve and the execution duration.


Let us recall that the transient 
market impact curve corresponds to 
restricting the market impact curve 
to time $t \le T$.   
In practice, we follow the guidelines of Section \ref{sec:principles} and,
following \eqref{eq:cal:avg1}, we compute
\begin{equation}
\label{eq:cal:avg2}
  {\hat \eta}_s=\EM{\epsilon(\omega)\Delta P_{sT}(\omega)}_{T=1}.
\end{equation}
Let us recall that 
$\EM{\ldots}_{T}$ means that averaging is performed after rescaling in time for each metaorder $\omega$  so that all durations correspond to $T=1$. So the function ${\hat \eta}_s$ is a function of the rescaled time $s$ (let us recall that in the paper we use the letter $s$ to refer to the rescaled time, whereas $t$ stands for the physical time). 
We sampled this estimation in $s \in [0,1]$ on 100 points using a uniform sampling grid.
Fig. \ref{fig:transient:all11} illustrates this computation and shows that a power law behavior
\begin{equation}
\label{betatr}
\hat \eta_{s}\propto s^{\gamma^{(tr)}}
\end{equation}
with $\gamma^{(tr)} = 0.64$ fits the curve according to a log-log regression when $s\le 1$ (the $(tr)$ subscript in $\gamma^{(tr)} = 0.68$ stands for {\em transient}). Our empirical findings are compatible whit Moro \textit{et al.} (\cite{citeulike:9771410} ) which found an exponent equal to $0.62$ for metaorders executed on London Stock Exchange (LSE) and $0.71$ for metaorders executed on Spanish Stock Exchange (BME, Bolsas y Mercados Espãnoles). 
\begin{figure}[!h]
	\centering
	\includegraphics[width=.8\linewidth]{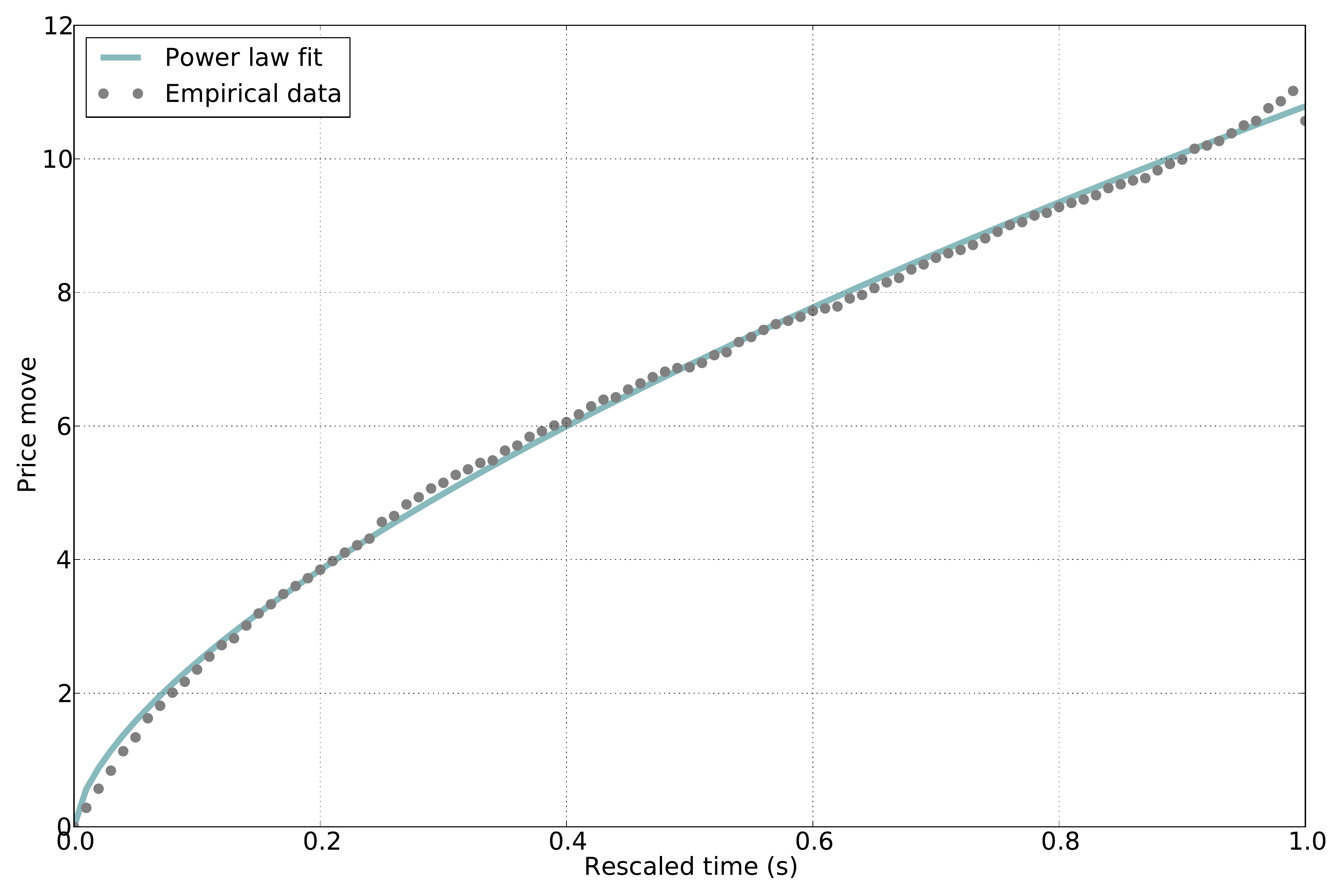}
	\caption{The estimated transient market impact curve $\hat \eta_{s}$ as defined by \eqref{eq:cal:avg2}; $\gamma^{(tr)}=0.64 $ .}
\label{fig:transient:all11}
\end{figure}

\subsection{Concavity and execution duration}

We shall now study the influence of the duration $T$ of the metaorder on the transient market impact using $\Omega^{(tr)}$. 
In order to avoid spurious effects of other explanatory variables, we chose to work on the subset of metaorders which correspond to a given range of daily participation  $R\in [1\%,3\%]$, selecting in this way  31,105 metaorders. 
We have checked that the so-obtained results do not qualitatively change for different intervals.
Following the averaging methodology of Section \ref{sec:principles} (see Eq. \eqref{eq:cal:avg} with $X(\omega) = T(\omega)$), we then fit the power $\gamma^{(tr)} $ for different durations intervals, to check the dependency with respect to $T$.
More precisely we choose the six following intervals (the duration are expressed in minutes) :  $T\in[3,15]$, $T\in[15,30]$, $T\in[30,60]$, $T\in[60,90]$, $T\in[90,300]$ and $T\in[300,510]$, each containing around 6,000   metaorder occurrences. Fig. \ref{fig:tmi} shows the transient market impact for each of these 6 groups. In order to point out the different regimes,  on each so-obtained graph, we performed a power-law fit.
The power-law fit is obtained by linear regression on a log-log representation. 
In order to test the robustness of our results we use bootstrap regressions 
drawing randomly 500 times 80\% of the available metaorders in the corresponding buckets  (see \cite{GIN97} for references on bootstrap). Table \ref{tab:boots:transient} gives the so-obtained statistics.
\begin{figure}[ht!]
     \begin{center}
        \subfigure [$T \in  [3,15[$, $\gamma^{(tr)} \simeq 0.80$]{%
            \label{fig:first_tmi}
            \includegraphics[width=0.33\linewidth]{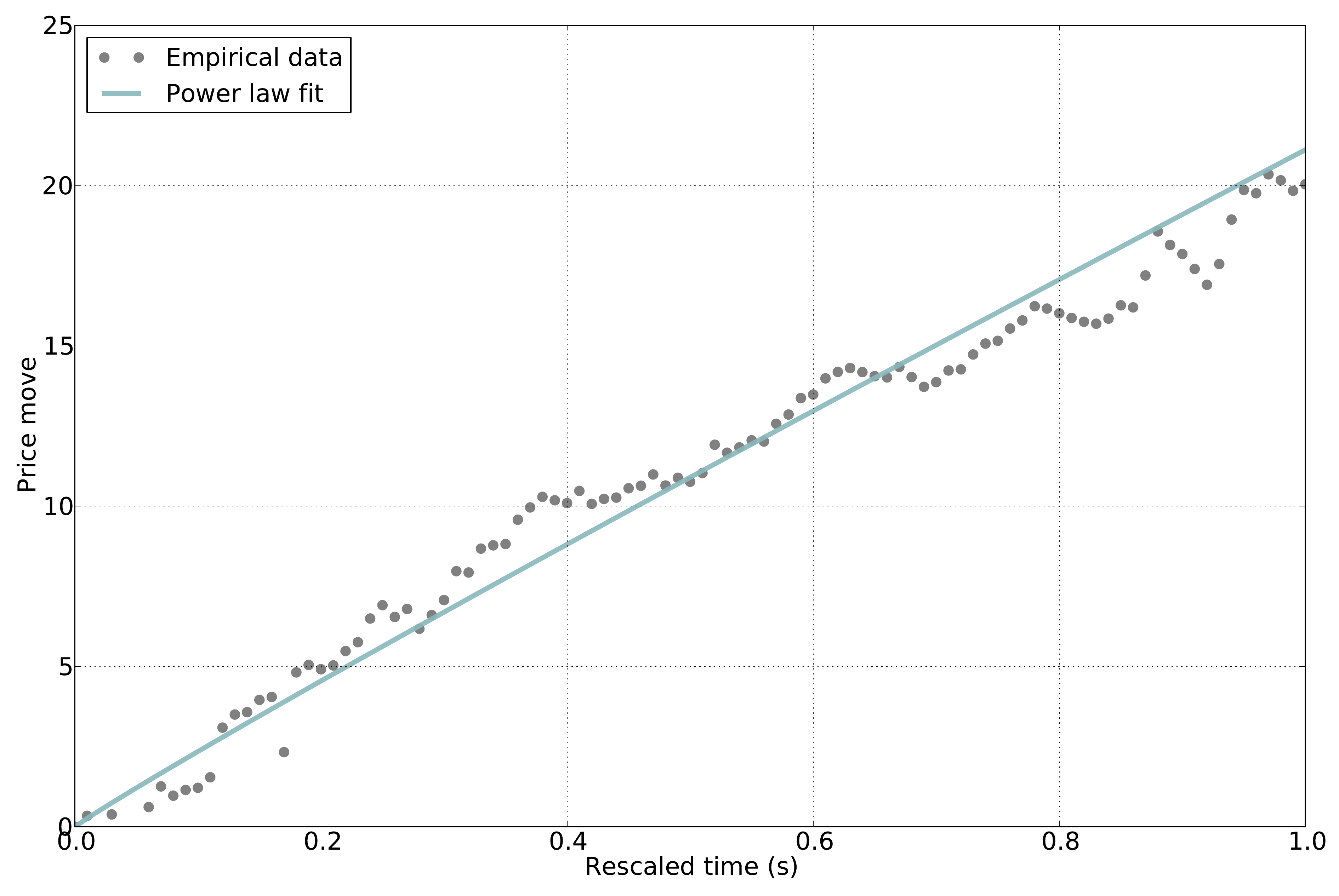}
        }%
        \subfigure [$T \in [15,30[$, $\gamma^{(tr)} \simeq 0.66$]{%
           \label{fig:second_tmi}
           \includegraphics[width=0.33\linewidth]{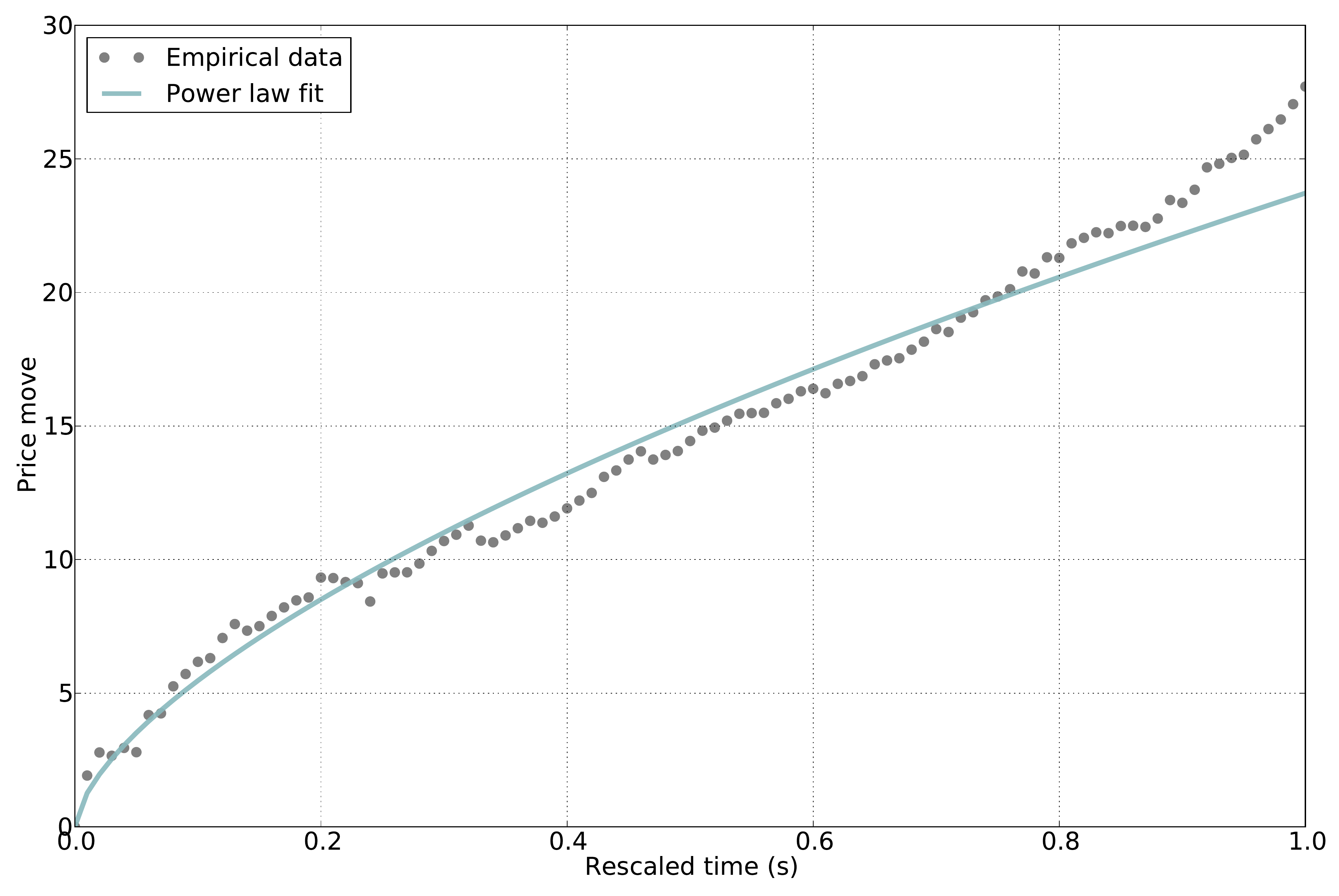}
}%
        \subfigure [$T \in [30,60[$, $\gamma^{(tr)} \simeq 0.63$]{%
           \label{fig:third_tmi}
           \includegraphics[width=0.33\linewidth]{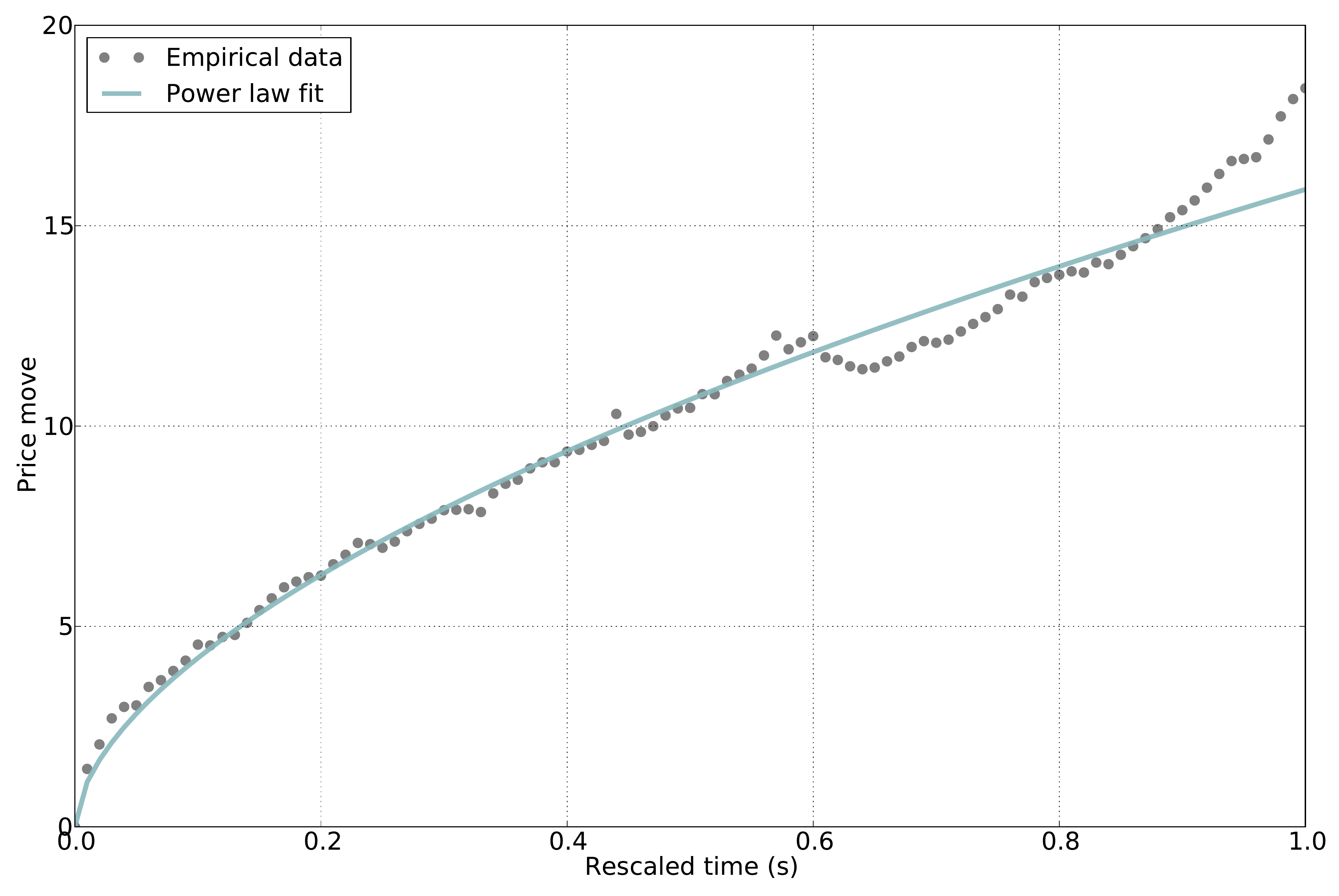}
        }\\ 
        \subfigure [$T \in  [60,90[$, $\gamma^{(tr)} \simeq 0.56$]{%
            \label{fig:forth_tmi}
            \includegraphics[width=0.33\linewidth]{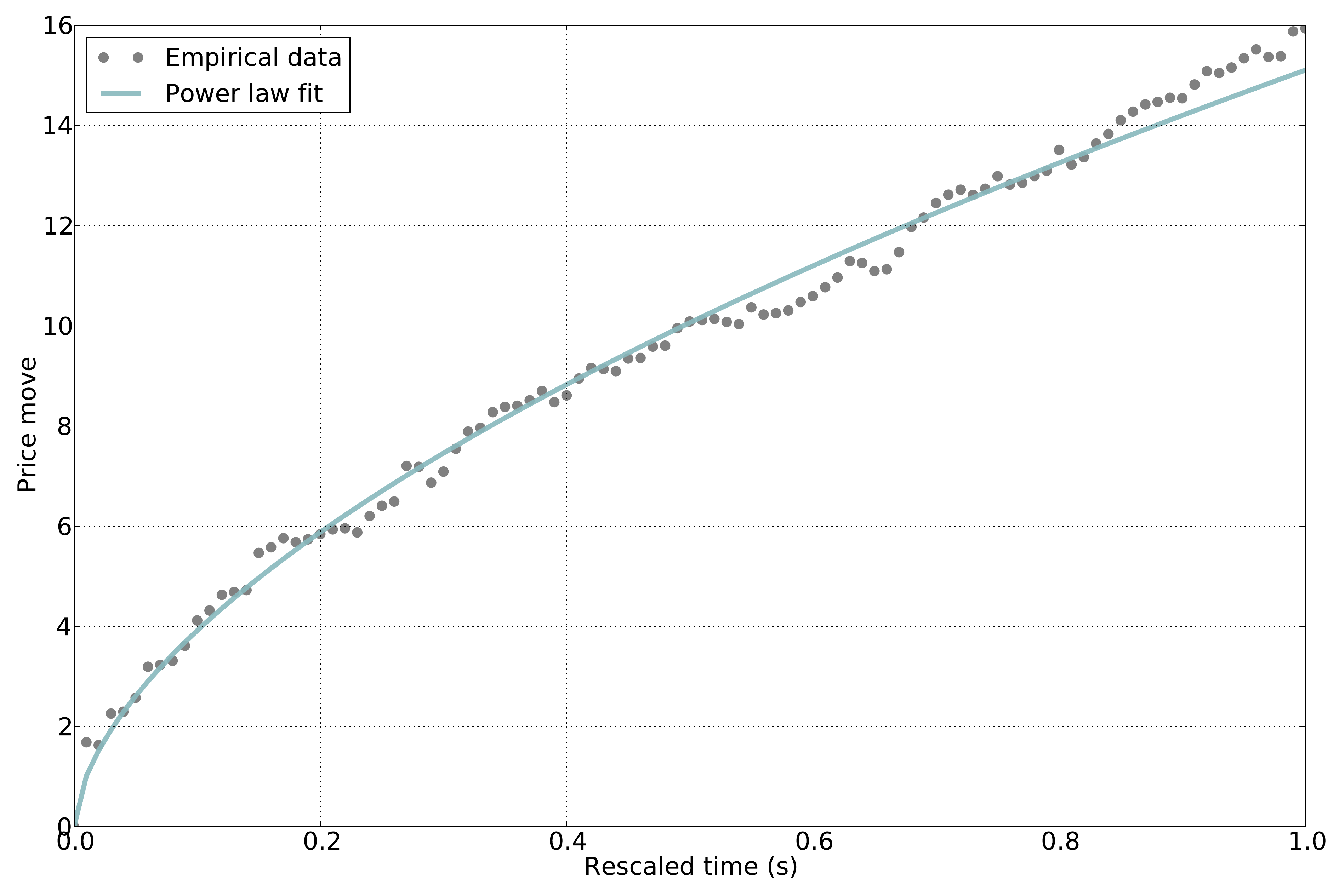}
        }%
        \subfigure [$T \in  [90,300[$, $\gamma^{(tr)} \simeq 0.54$]{%
            \label{fig:fifth_tmi}
            \includegraphics[width=0.33\linewidth]{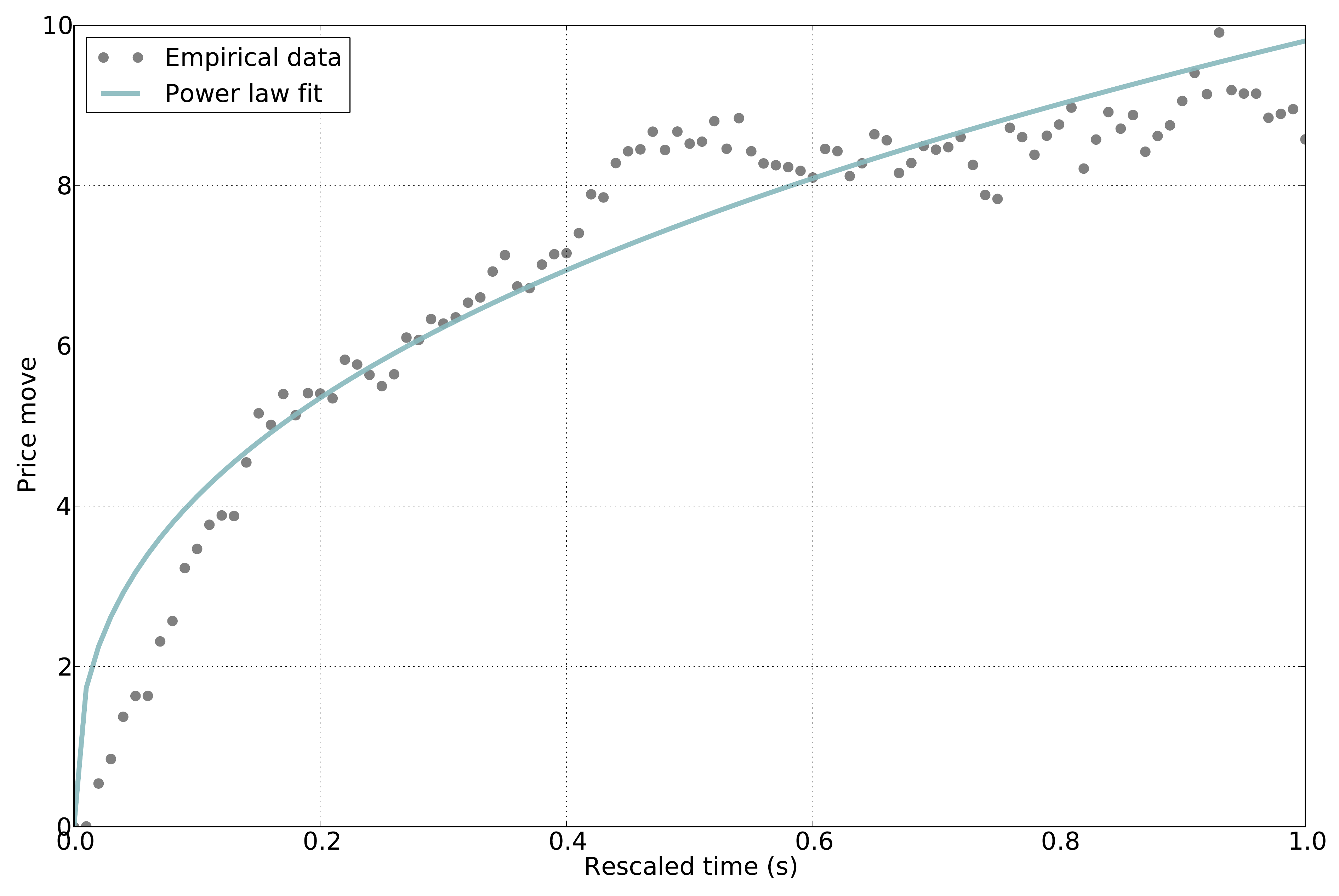}
        }%
        \subfigure [$T \in I_T = [300,+\infty[$]{%
            \label{fig:sixth_tmi}
            \includegraphics[width=0.33\linewidth]{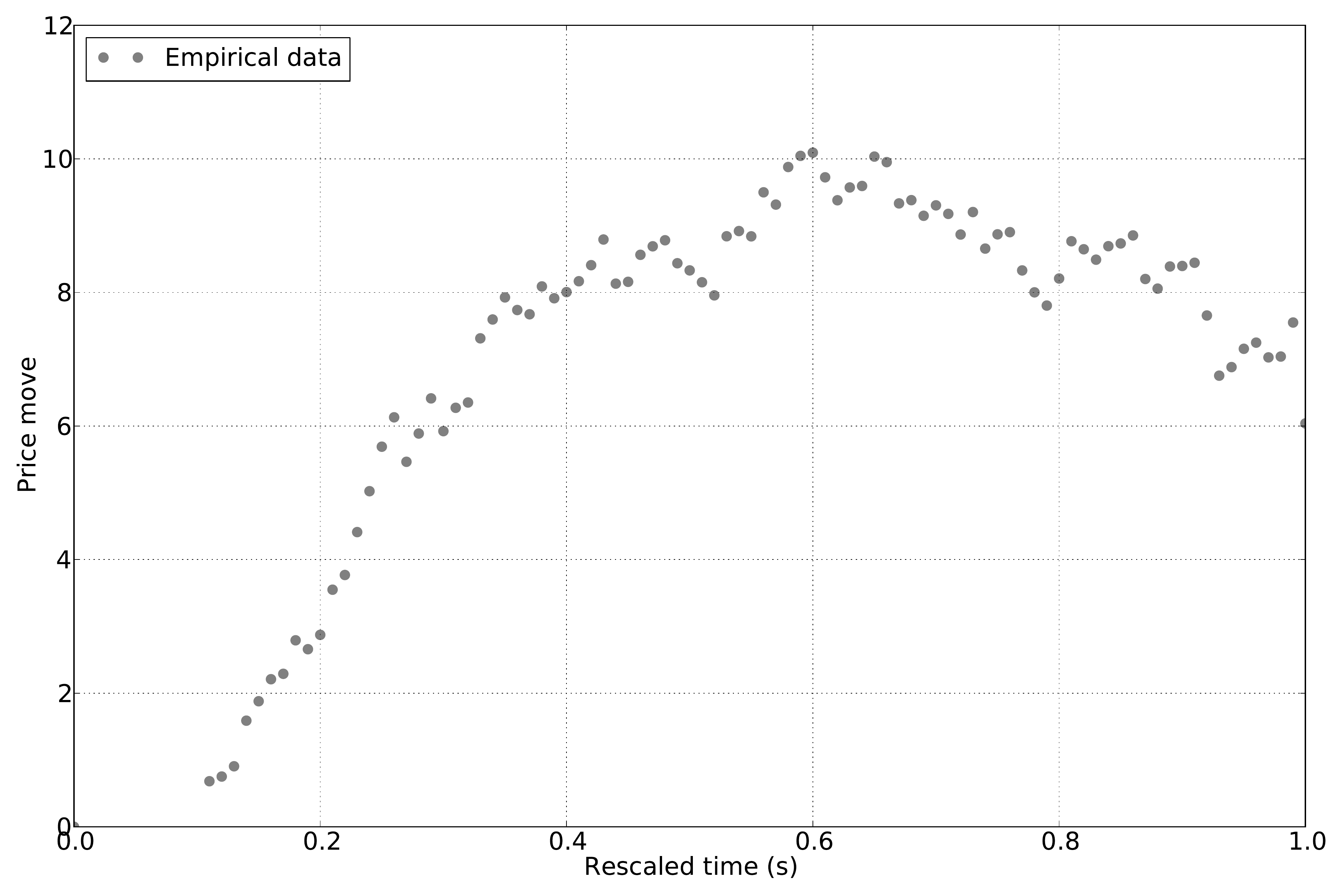}
        }%
    \end{center}
    \caption{
        Estimated transient market impact as a function of the renormalized time $s$ for a daily participation $R \in [1\%,3\%]$ and for different duration intervals.
        Each graph corresponds to a different duration interval.
        In each case, 
        a power-law is fitted leading to an estimation of $\gamma^{(tr)}$.
        Statistics on the corresponding distributions of $\gamma^{(tr)}$ are shown in Table \ref{tab:boots:transient}.
        On (a)-(e), we see that the larger $T$ the larger the curvature of the transient market impact and the smaller the temporary market impact. For very long metaorders (f), decay seems to happen before end of market order.
     }%
   \label{fig:tmi}
\end{figure}

\begin{table}[h!]
  \centering
  \begin{tabular}{|c|l|c|c|c|c|c|c|} \hline
    Subset name & Execution time & Average $\gamma^{(tr)}$  & Q5$\%$ & Q25$\%$ & Q25$\%$ & Q75$\%$ & Q95$\%$  \\ \hline\hline
${\calT}_1$ & $T=[3,15]$ & 0.80 & 0.76 & 0.78 & 0.80 & 0.82 & 0.85 \\\hline
${\calT}_2$ & $T=[15,30]$ & 0.66 & 0.62 & 0.65 & 0.66 & 0.68 & 0.70 \\\hline
${\calT}_3$ & $T=[30,60]$ & 0.62 & 0.58 & 0.60 & 0.62 & 0.64 & 0.66 \\\hline
${\calT}_4$ & $T=[60,90]$ & 0.55 & 0.49 & 0.52 & 0.56 & 0.58 & 0.62 \\\hline
${\calT}_5$ & $T=[90,300]$ & 0.54 & 0.48 & 0.52 & 0.55 & 0.57 & 0.62 \\\hline
  \end{tabular}
  \caption{Statistics (mean and quantiles) on the distribution of the power-law exponent $\gamma^{(tr)}$ of the transient market impact estimation of metaorders with a participation rate 
  $R \in [1\%,3\%]$. The exponent is estimated using log-log regression conditioned on different duration intervals. The larger $T$ the larger the curvature of the transient market impact and the smaller the temporary market impact (see Fig. \ref{fig:tmi}).}
  \label{tab:boots:transient}
\end{table}
We observe that transient market impact, considered in renormalized time $s=t/T$ is a multi-regime process. The first five plots of Fig. \ref{fig:tmi} (from (a) to (e)) show clearly that, for a daily participation rate (as we already mentioned changing the participation rate interval does not affect the results),  when the duration of a metaorder decreases
\begin{itemize}
\item the transient market impact of a metaorder increases and
\item the curvature decreases leading to an almost linear transient market impact for small durations.
\end{itemize}
Thus, when executed faster, a metaorder seems to have a stronger and more linear impact. These results are rather intuitive : when a metaorder has a short duration, the market has hardly the time to "digest" it resulting into a strong linear impact (not enough time for relaxation). 
However, 
Fig. \ref{fig:tmi}(e) seems to show that a kind saturation is reached before the end of the metaorder. Actually Fig. \ref{fig:tmi}(f) surprisingly shows that when the duration $T$ becomes very large, the market impact curve starts decaying {\em before} the end of the metaorder.  From our knowledge, this is the first study pointing out this effect that we cannot explain at this stage.

\paragraph{Market prediction of metaorder sizes.}
In this paragraph, we want to study whether the market has or has no precise insights about the total size of a given metaorder before the end of its execution (apart of course from the unconditionnal distribution of the metaorder sizes).
We also document the shape of the transient market impact in absolute time $t=sT$.
In order to do so, we compare the transient market impact of different metaorders which have the same \emph{average participation rate} $\dot\nu=R/T$.
It is interesting to check if two metaorders $\omega_1$ and $\omega_2$ with similar $\dot\nu$ but having different durations $T_1$ and $T_2$ have the same transient impact as far as $t\le T_1\wedge T_2$.

For that purpose, we create 
five groups of meta-orders $A_i$ ($i=1,\dots,5$) with the same \emph{average participation rate } $\dot v=R/T$ but different durations. 
In detail:
\begin{equation}
\label{def:exec_speed}
A_i=\left\{\omega \in \Omega^{(tr)}: R(\omega)\in[2^{i-1}R_0, 2^i R_0[ \mbox{~~and~~} T(\omega)\in[2^{i-1}T_0, 2^iT_0[\right\},
\end{equation}
where $R_0 = 0.25$ and $T_0 = 5$ seconds.
Thus, all the selected meta-orders correspond, in 
a good approximation, to the same average participation rate
$\dot v = R_0/T_0 = 0.05 s^{-1}$.  Moreover increasing the index $i$ by one corresponds to double metaorders duration. 

For each group $A_i$, we compute the corresponding transient market impact (with rescaled time) $\hat \eta_s^{(i)}$.
For each $i=1,\dots,4$, Fig.  \ref{fig:exec_speed} shows $\hat \eta_{s}^{(i)}$ with the first half of $\hat \eta_{2s}^{(i+1)}$ for $s\leq 1$. One can see that, in each of the subplots, the two market impact curves are very close, indicating that the market basically does not anticipate the size of the corresponding meta-orders. Let us point out that the same results would be obtained when changing $R_0$ and/or $T_0$.

\begin{figure}[ht!]
     \begin{center}
        \subfigure [$A_1$ and $A_2$]{%
            \label{fig:first_exec_speed}
            \includegraphics[width=0.4\linewidth]{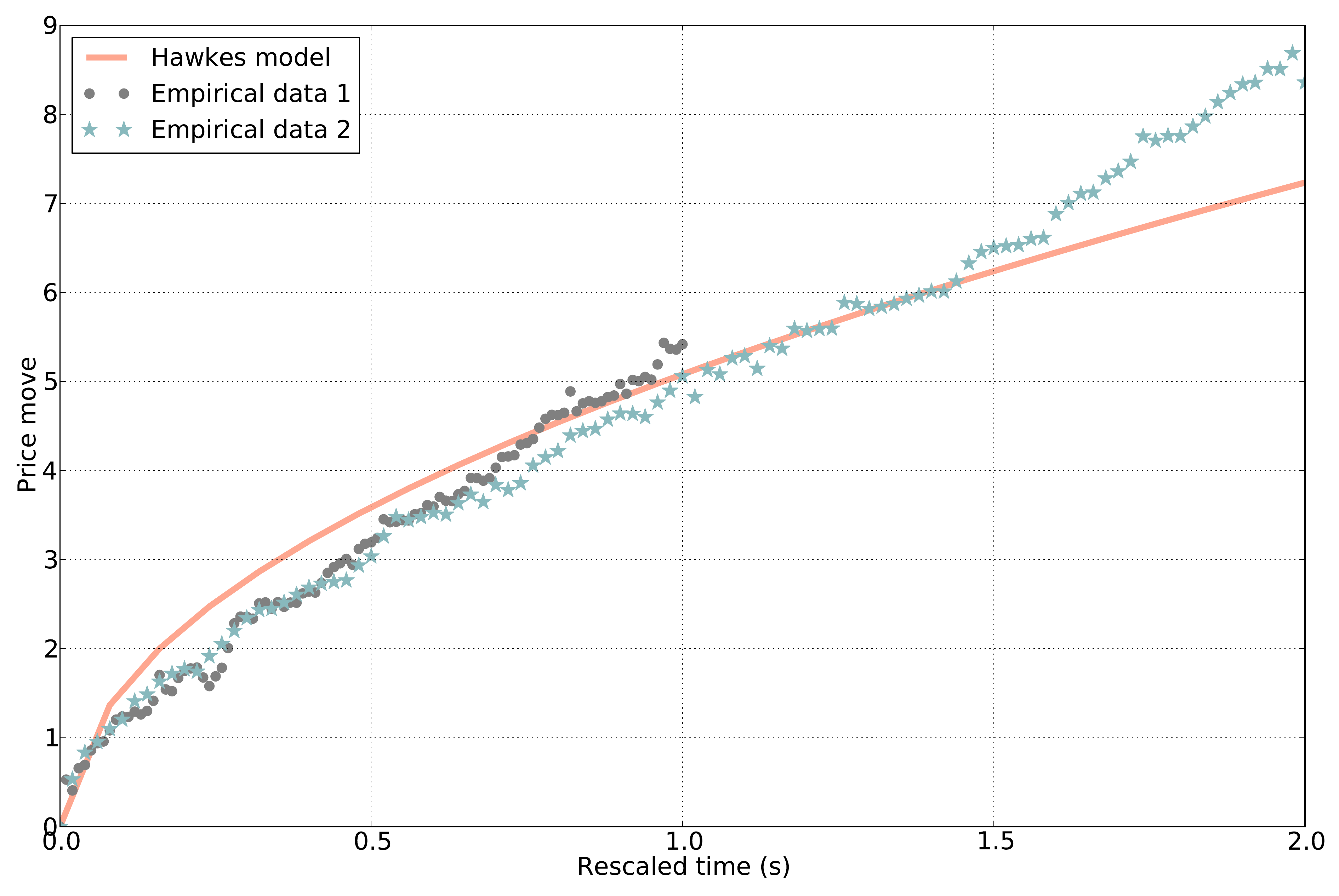}
        }%
        \subfigure[$A_2$ and $A_3$]{%
           \label{fig:second_exec_speed}
           \includegraphics[width=0.4\linewidth]{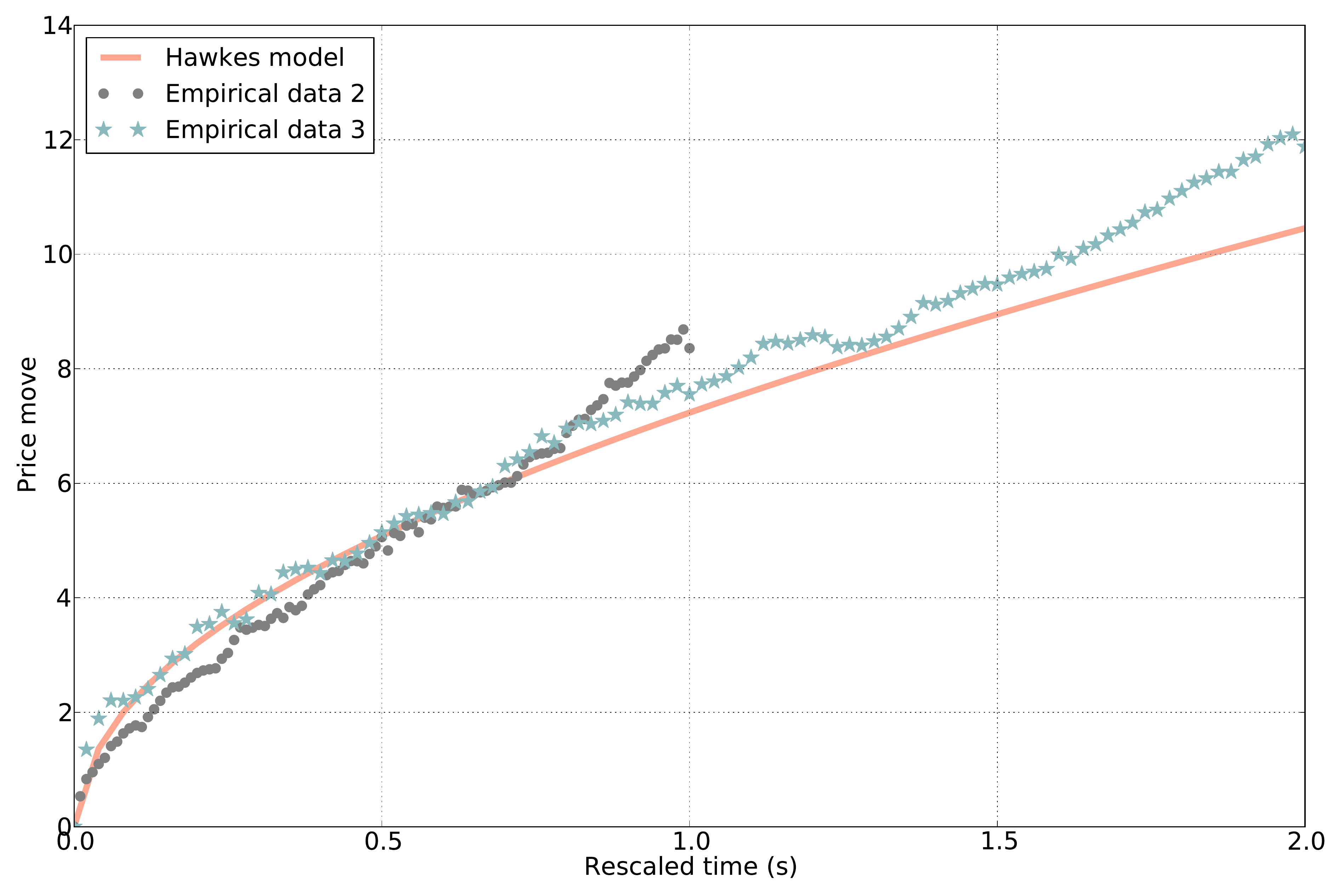}
        }\\ 
        \subfigure[$A_3$ and $A_4$]{%
            \label{fig:third_exec_speed}
            \includegraphics[width=0.4\linewidth]{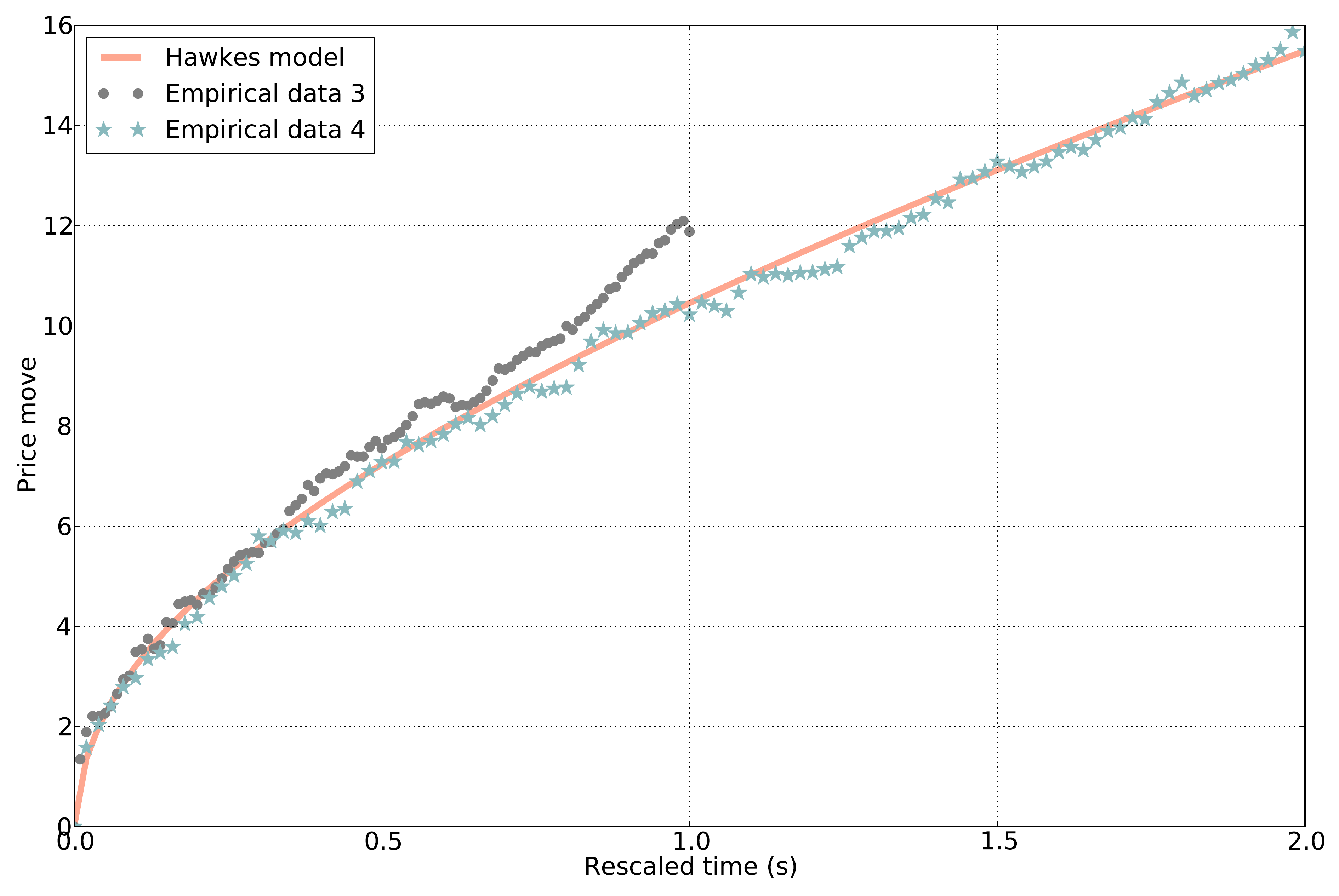}
        }%
 \subfigure[$A_4$ and $A_5$]{%
            \label{fig:forth_exec_speed}
            \includegraphics[width=0.4\linewidth]{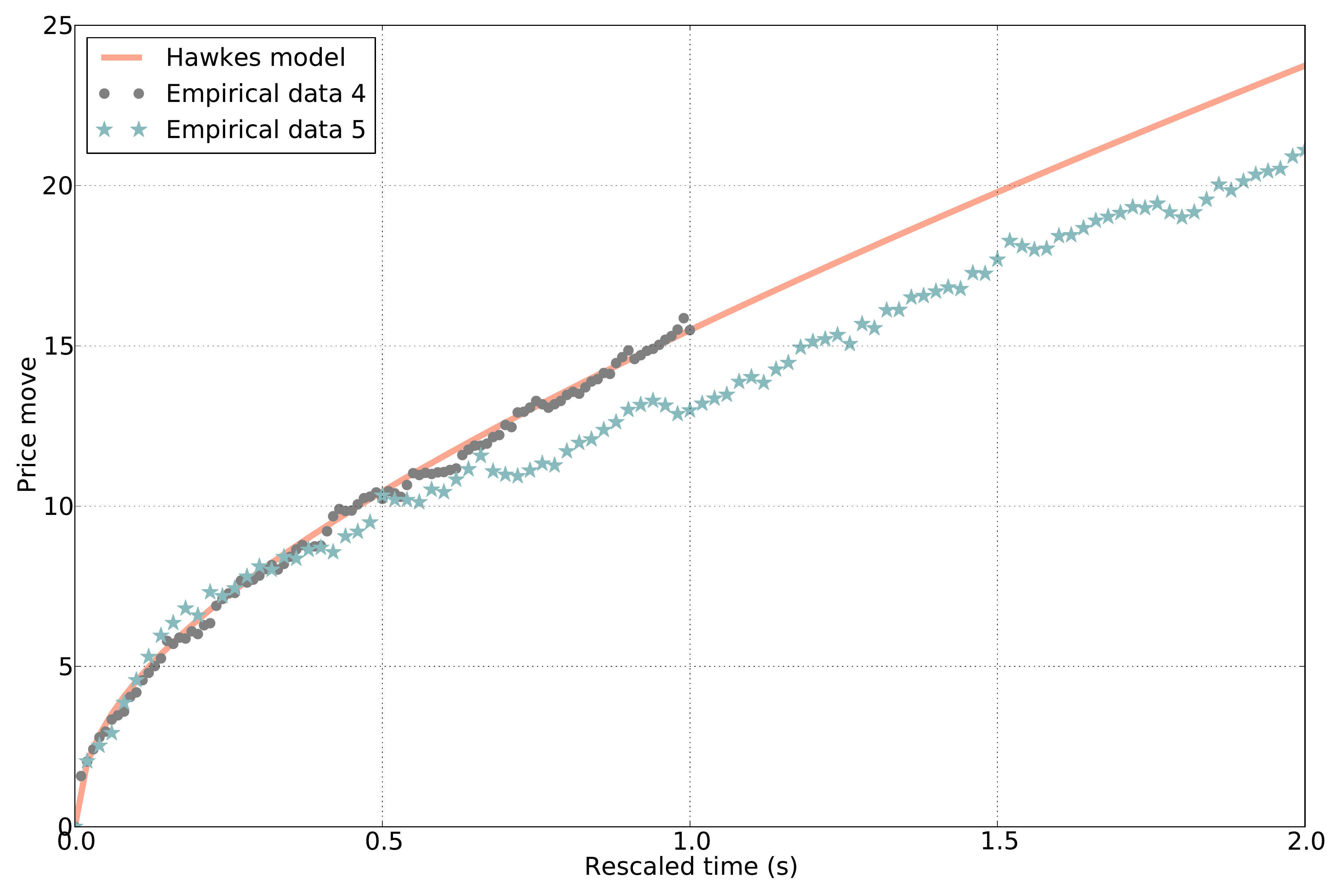}
        }%
    \end{center}
    \caption{On each subplot two transient market impact
  $\hat \eta_{s}^{(i)}$ and $\hat \eta_{2s}^{(i+1)}$ curves for $s \in [0,1]$
  are displayed. They correspond to two duration intervals (one being twice the other, see \eqref{def:exec_speed}). 
The fact that 
 the two curves are very close, indicates that the market basically does not anticipate the metaorder size.
  Top-left (resp. top-right) subplot corresponds to $i=1$ (resp. $i=2$) and     bottom-left (resp. bottom-right) subplot corresponds to $i=3$ (resp. $i=4$). 
     }%
   \label{fig:exec_speed}
\end{figure}

\section{The decay market impact curve}
\label{sec:decay}
\subsection{A well known stylized fact : convexity of decay market impact}

During the execution of the metaorder the price is pushed in the adverse direction making it less attractive as time goes by reaching higher level (temporary impact) at the end of the execution. After the execution a reversal effect is expected as seen in Fig. \ref{fig:prox}. This is the decay part or relaxation of the market impact.
Some qualitative explanations can be sketches. 
For instance, consider the whole market impact cycle (transient, temporary and decay phases) in the spirit of a continuous version of Kyle's model \cite{citeulike:3320208}: a \emph{stylized market maker} makes a price from 0 to $T$ and unwinds its accumulated inventory after $T$ at a price compatible with the risk he usually take to accept the position.
Another potential explanation is mechanical: if the orderbooks refill with a constant rate independent of the metaorder, the metaorder will first consume the limit orderbook on one side, pushing the price in its direction, and then let an orderbook imbalanced enough at $T$ such that symmetric and random consumption of liquidity on the two sides of the books will implement a mechanical decay. It is typically supported by a \emph{latent orderbook model} like the one developed in \cite{citeulike:12802757}.
Last but not least, a model of \emph{ impact plus decay at an atomic size} (i.e. for each child order generated by the metaorder, like in \cite{citeulike:13143978}) will generate a transient phase as soon as the decay does not end when the next child order is generated (this will typically be the case for power law decay); after $T$, the ``cumulated decay'' will express itself, generating an observable reversal.


The existing empirical literature of decay metaorders market impact is limited (\cite{citeulike:9771410}, 
\cite{BR}) since the difficulty of obtaining data is very high. In the first study, Moro \textit{et al.} are the first showing a decay of the impact to a level roughly equal to $0.5\sim 0.7$ of its highest peak. In the second study, Bershova and Rakhlin shows the decay is a two-regime process: slow initial power decay followed by a faster relaxation.

In this section we confirm that the transient market impact curve is convex and that it seems to have a slow initial regime.

\subsection{Numerical results.}

To have a chance to observe intraday decay, it is needed to restrict the numerical study to orders ending long enough before the close.
We chose to follow  \cite{citeulike:9771410} and selected metaorders ending before the close, i.e. such that $t_0(\omega)+2T(\omega)$ takes place before closing time.
This exactly corresponds to the database $\Omega^{(de)}$ defined in Table \ref{tab:filters}.

Following the same lines as in the previous section, we use an averaging methodology to compute $\hat \eta_{s}$ for $s\le 2$ for different intervals of duration (as previously, $s$ corresponds here to the rescaled time, so $s=2$, corresponds to the physical time $t=2T$).
Fig. \ref{fig:subfigures} shows such estimations (of both transient and market impact decay curves) for first four  intervals used in Fig. \ref{fig:tmi}. Thus the transient parts of Fig. \ref{fig:subfigures}(a-d) are respectively exactly the same as the curves displayed in Fig. \ref{fig:tmi}(a-d).

\begin{figure}[ht!]
     \begin{center}
        \subfigure []{%
            \label{fig:first}
            \includegraphics[width=0.4\linewidth]{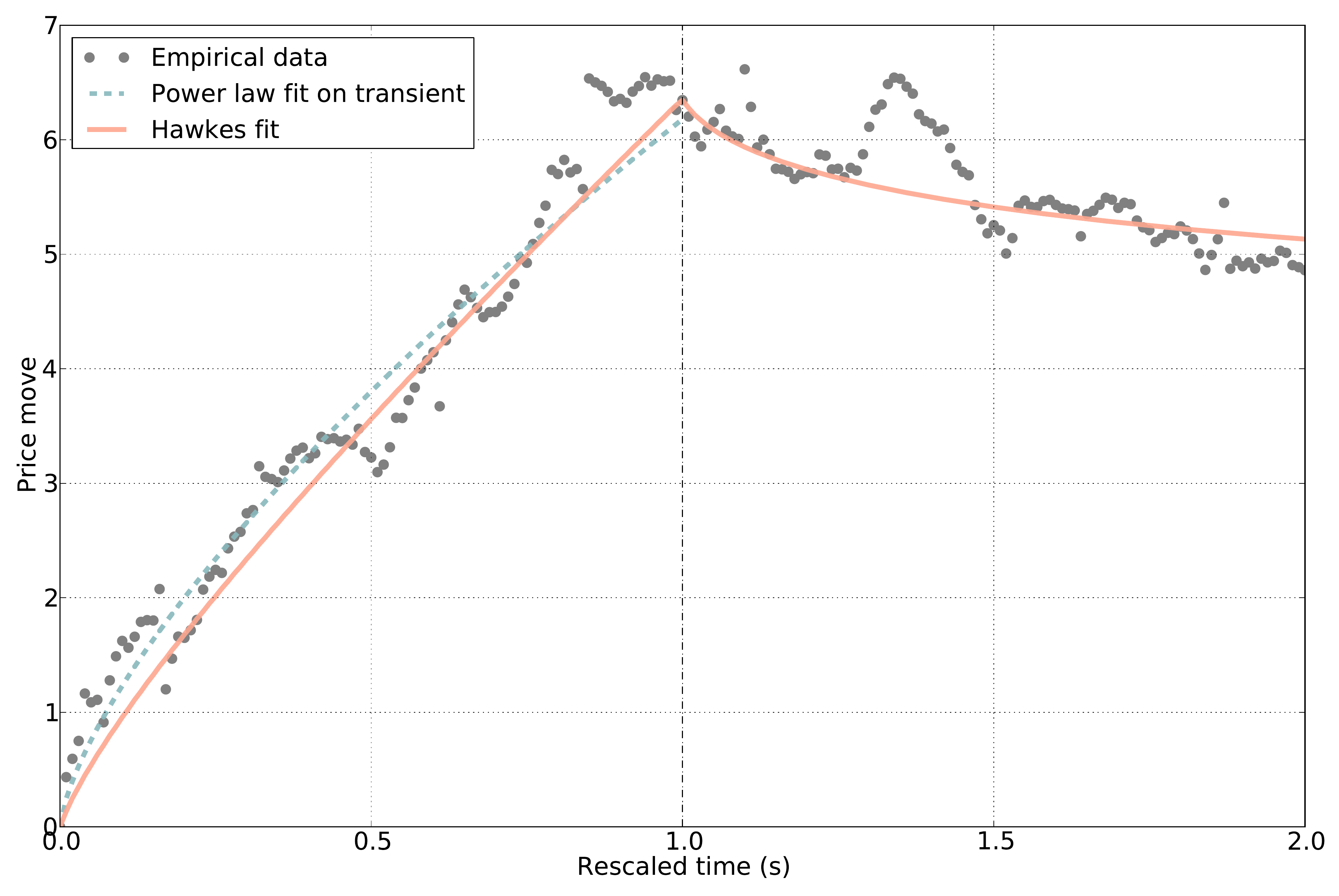}
        }%
        \subfigure[]{%
           \label{fig:second}
           \includegraphics[width=0.4\linewidth]{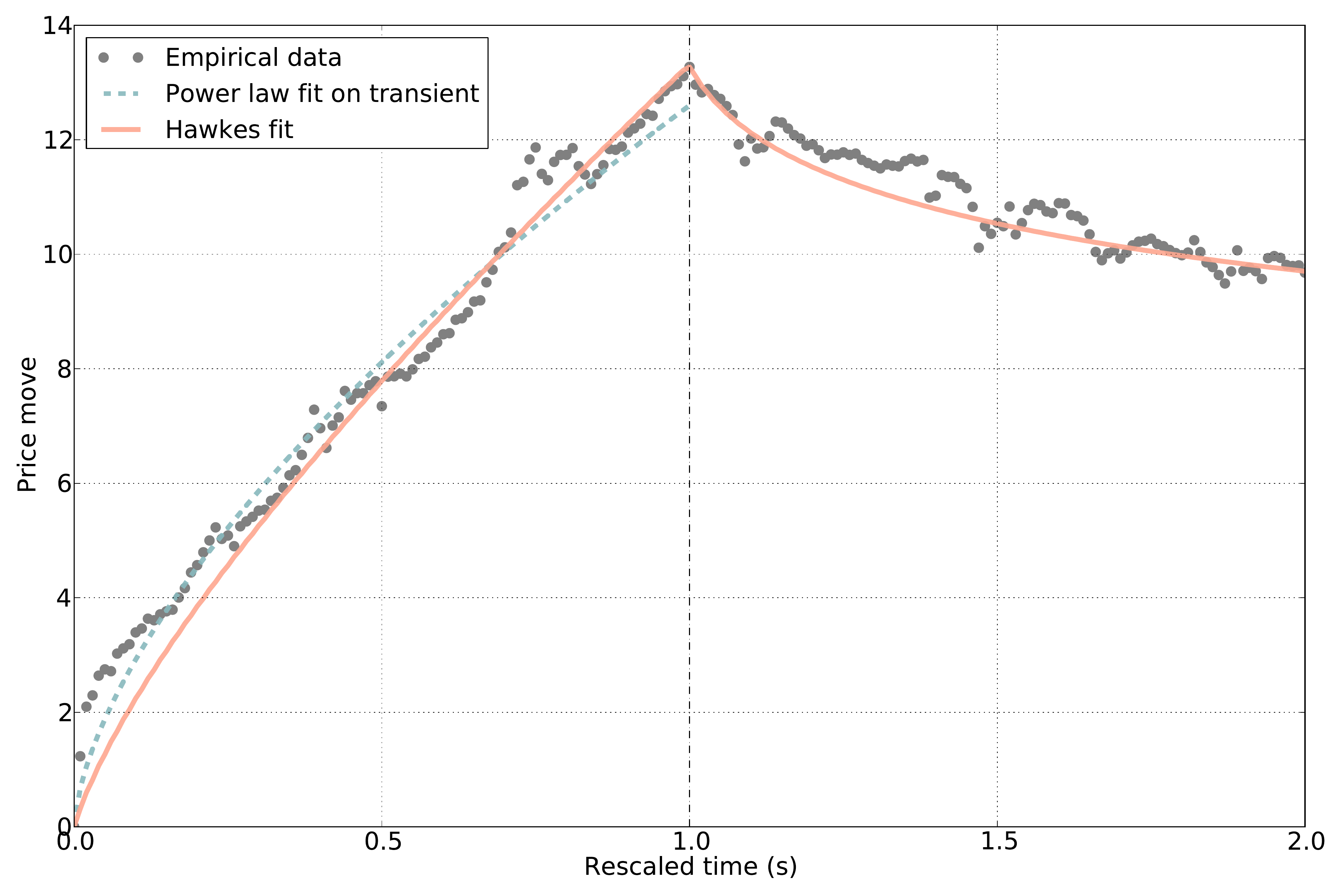}
        }\\ 
        \subfigure[]{%
            \label{fig:third}
            \includegraphics[width=0.4\linewidth]{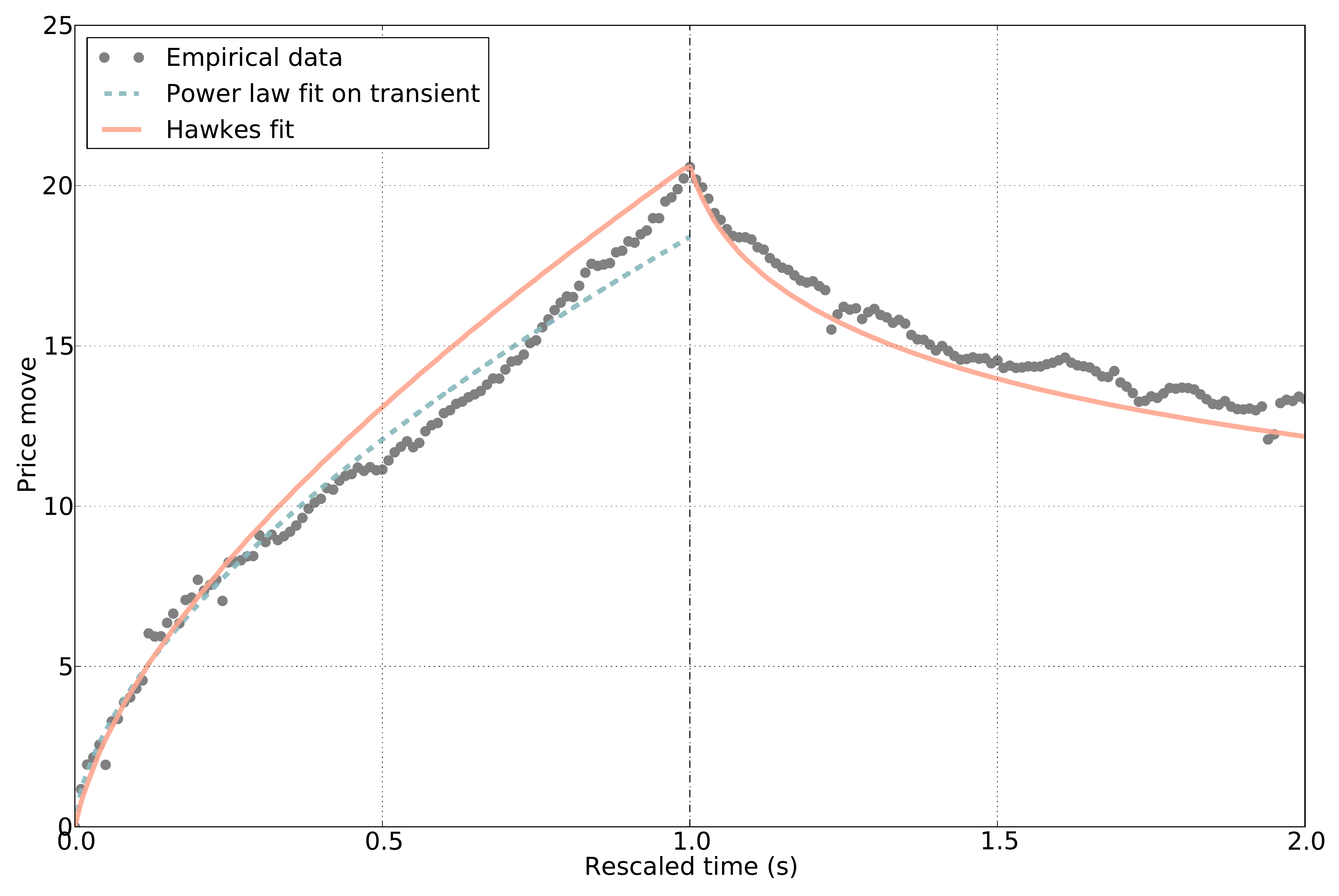}
        }%
 \subfigure[]{%
            \label{fig:fourth}
            \includegraphics[width=0.4\linewidth]{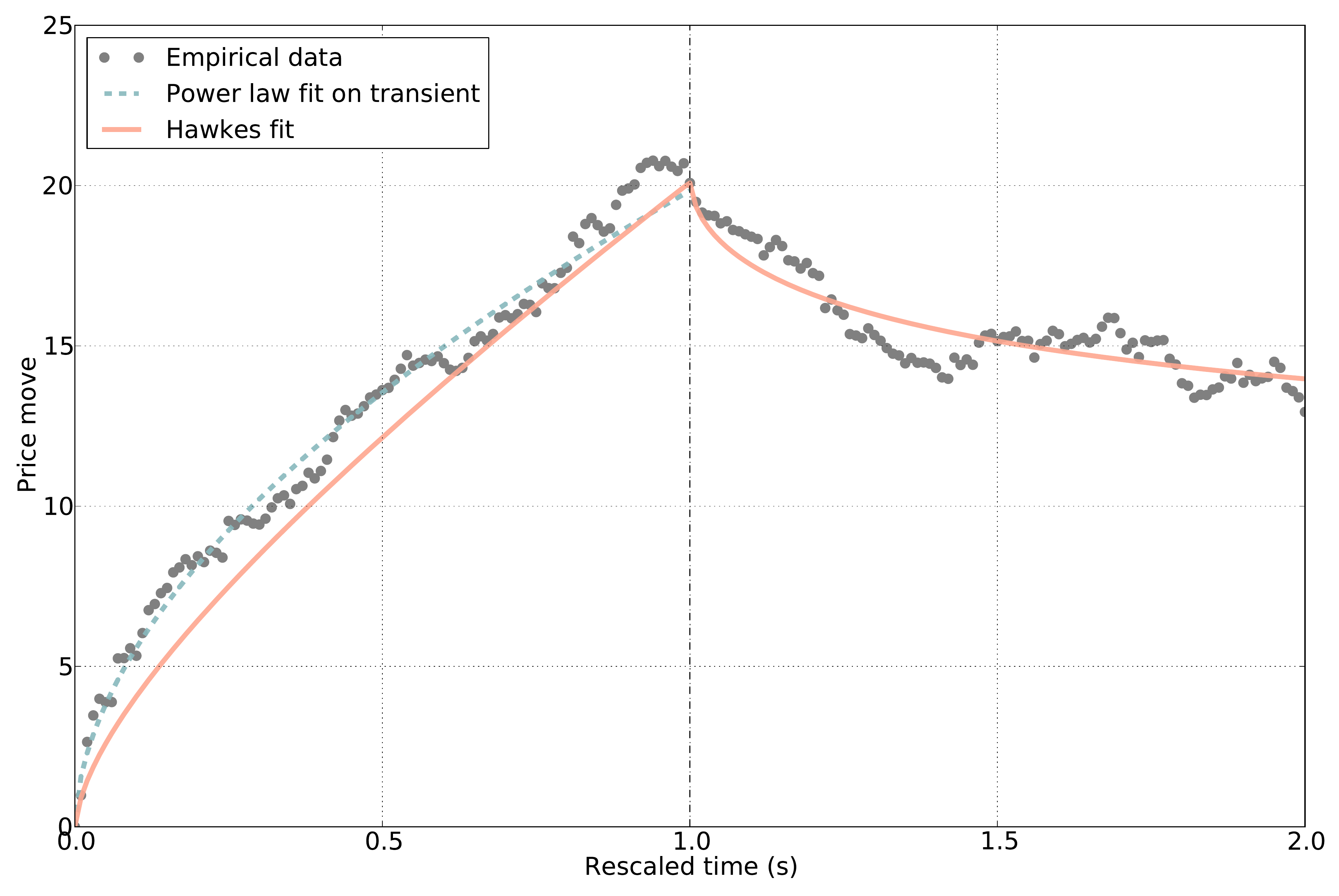}
        }%
    \end{center}
    \caption{%
Transient and market impact decay curve estimations $\hat \eta_{s}$ for  $R\in [1\%,3\%]$ and for different $T$ ranges
(same as in Fig. \ref{fig:tmi}(a-d)). Power-law fit of the transient part are shown.  Fits with the HIM model (see Section \ref{sec:HIM}) are also displayed. (a) $T \in [3,15[$, (b) $T \in  [15,30[$, (c) $T \in  [30,60[$ and (d)   $T \in [60,90[$.
      }%
   \label{fig:subfigures}
\end{figure}
Log-log plots of the corresponding market impact decays are shown on Fig. \ref{fig:loglogdecay}. 
More precisely, in order to study the decay rate (towards the permanent impact value), we displayed log-log plots of 
 $\hat \eta_{s}-\hat \eta_{s=2}$ as functions of $s-1$ for $s \in ]1,2]$. They clearly show that the decay is much slower at the very beginning (i.e., right after the end of the execution of the metaorder). We have checked that changing the daily participation $R$ does not affect qualitatively the result. This result confirms the ones obtained previously by \cite{citeulike:12838207} and \cite{citeulike:12825932}.

\begin{figure}[ht!]
     \begin{center}
        \subfigure[]{%
            \label{fig:decfourth}
            \includegraphics[width=0.4\linewidth]{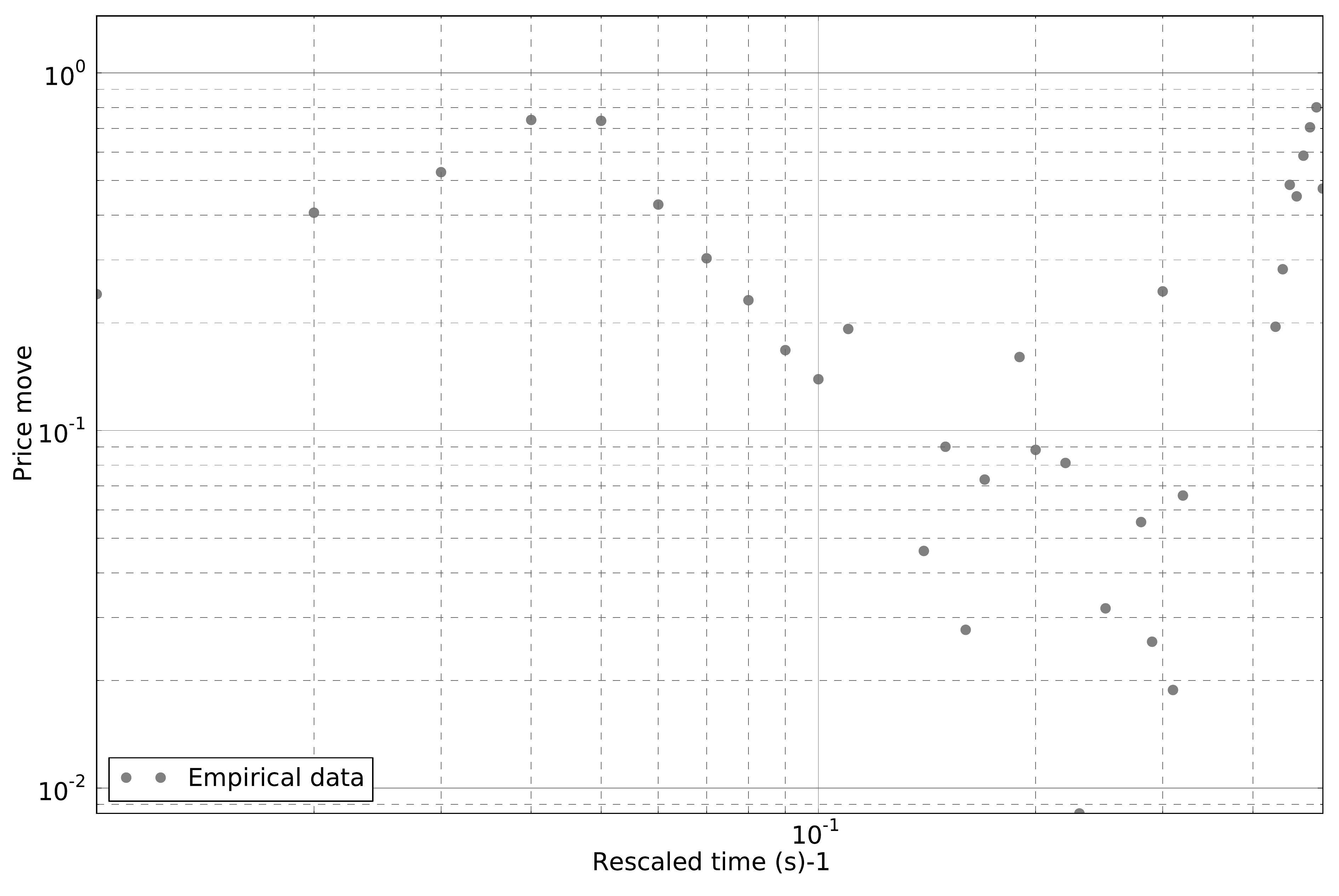}
        }%
        \subfigure []{%
            \label{fig:decfirst}
            \includegraphics[width=0.4\linewidth]{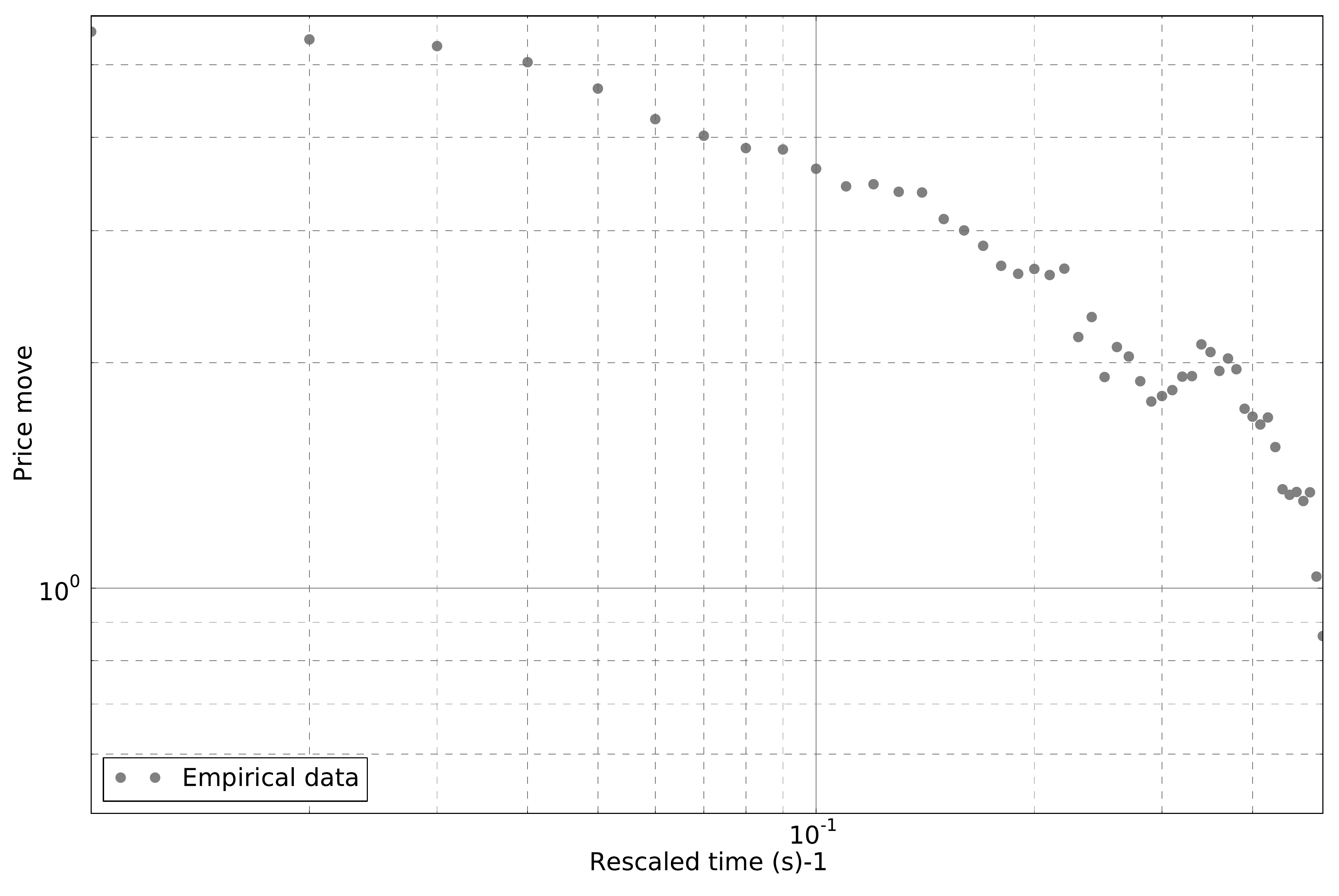}
        }\\ 
        \subfigure[]{%
           \label{fig:decsecond} 
           \includegraphics[width=0.4\linewidth]{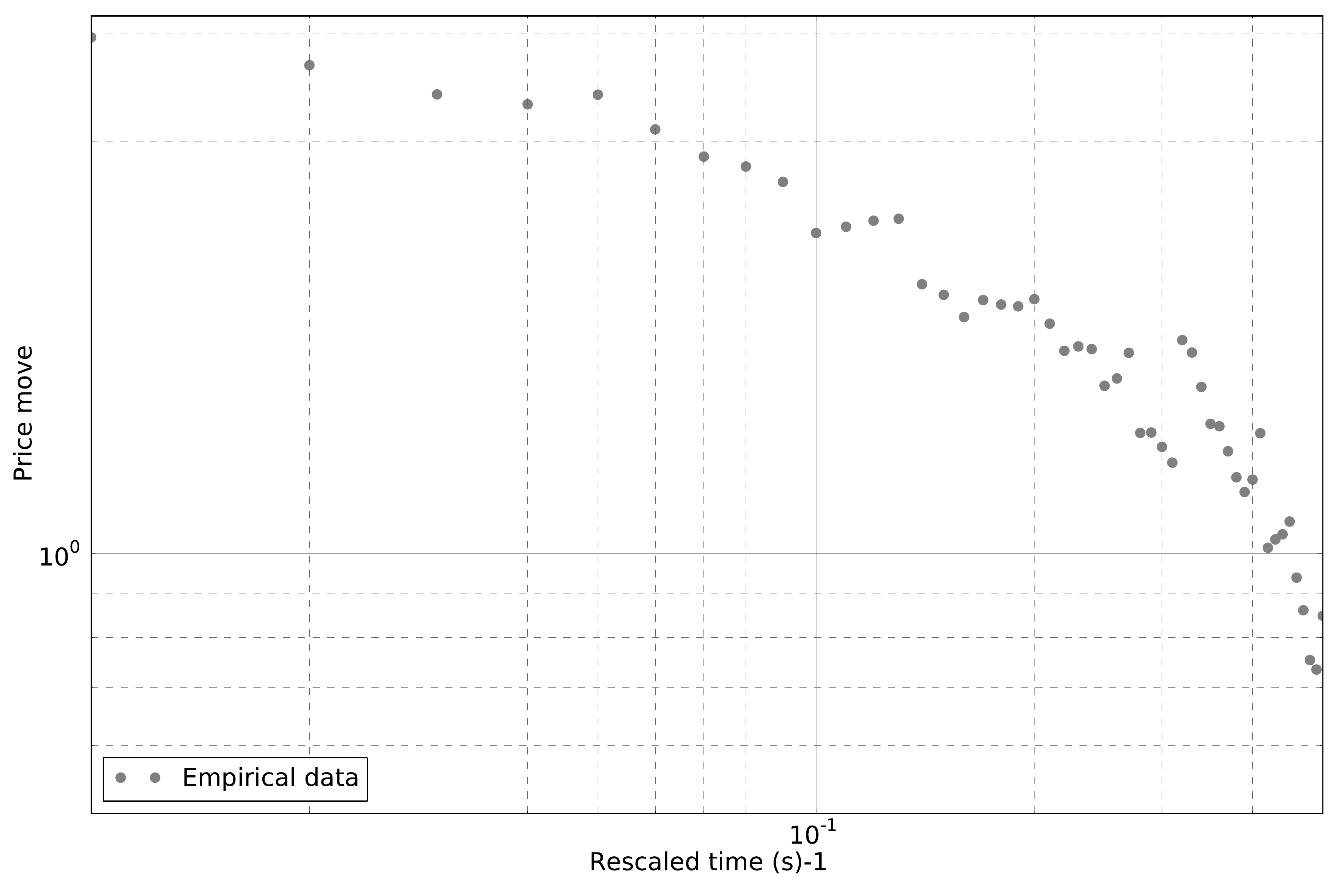}
        }%
        \subfigure[]{%
            \label{fig:decthird}
            \includegraphics[width=0.4\linewidth]{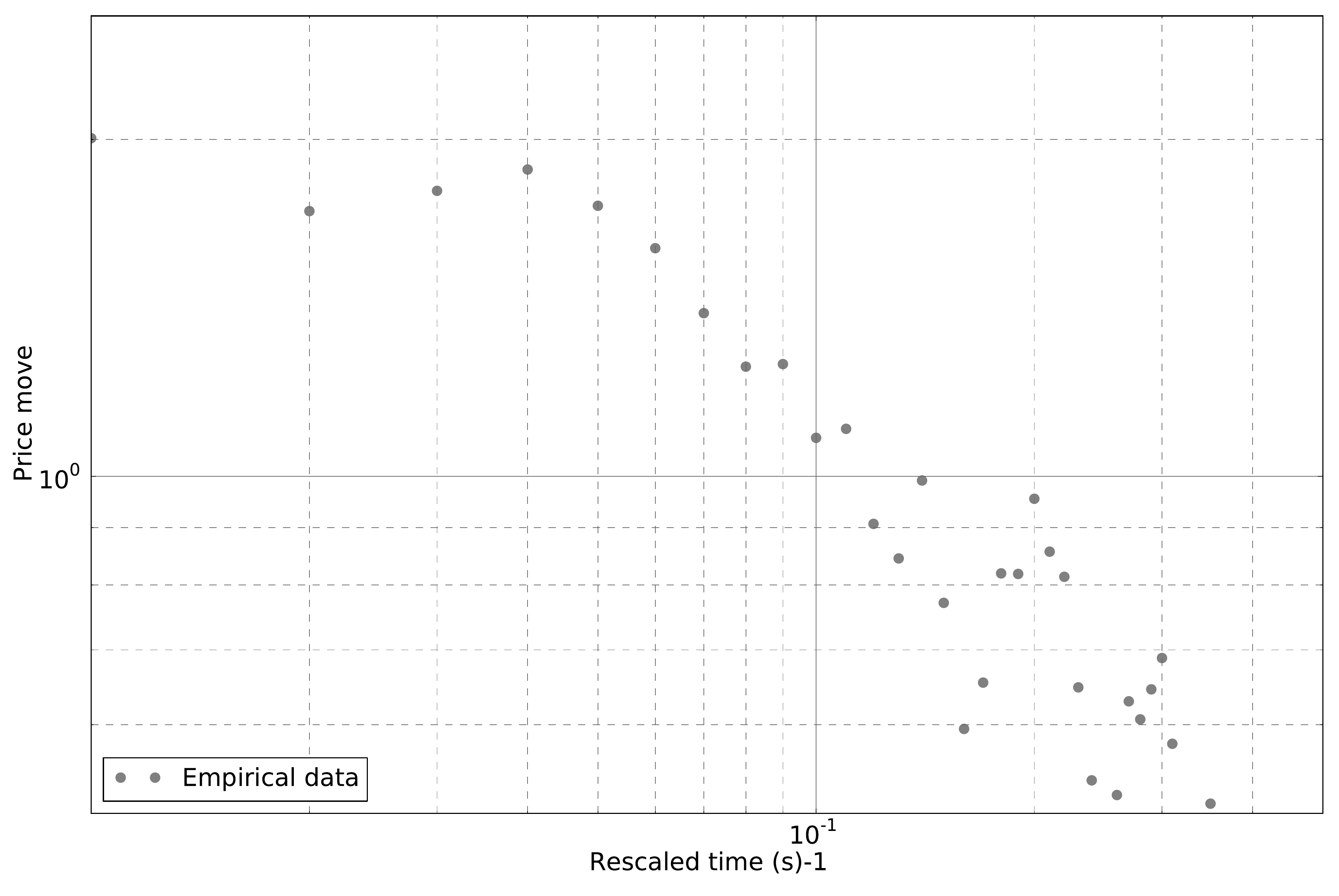}
        }%
    \end{center}
    \caption{%
Log-log plots of market impact decay curve estimations. More precisely, it shows the log-log plot of 
$\hat \eta_{s}-\hat \eta_{s=2}$ as a function of $s-1$ for $s \in ]1,2]$, 
for $R \in [1\%,3\%]$ and for several $T$ ranges (same as in Fig. \ref{fig:tmi}(a-d)). It shows that the decay is  slower at the very beginning (for $s$ close to 1). (a) $T \in [3,15[$, (b) $T \in [15,30[$, (c) $T \in  [30,60[$ and (d) $T \in [60,90[$.
      }%
   \label{fig:loglogdecay}
\end{figure}


\section{Transience and decays in the Hawkes Impact Model toy model}
\label{sec:HIM}
\subsection{Hawkes based models for microstructure}
Hawkes processes have already been proved successful for modeling high frequency financial time-series (see \cite{citeulike:9217792, citeulike:3856893, BM, Hewlett1}) and in optimal trading schemes (see \cite{Alfonsi-Blanc}). Hawkes processes are point processes with a stochastic intensity which  depends on the past of the process.

Following \cite{citeulike:9217792}, we consider the following  price model. Let $P_t$ be a proxy for the high-frequency price of an asset (e.g., last-traded price, mid-price, \ldots). For the sake of simplicity, we shall not consider the size of the jumps in the price and consider that they are only of size $1$. Let $(J^+_t$, $J^-_t)$ be the point processes representing respectively upward and downward jumps of $P_t$.
\begin{equation}
\label{hawkes_price_definition}
P_t = J^+_t-J^-_t.
\end{equation}
Let $\lambda^+$ and $\lambda^-$ the respective intensities of $(J^+_t$, $J^-_t)$.
It is well known that at microstructure level, the price is highly mean reverting (at least for large tick-size assets). It has been shown (\cite{citeulike:9217792}) that this mean-reversion property is well mimicked using a $2$-dimensional Hawkes process using only a single "cross" kernel $\varphi(t)$ :
\begin{equation}
\label{hawkes_intensities_definition}
\lambda^+_t= \mu + \varphi \star dJ^-_s~~\mbox{and}~~{\mbox{and}}~~
\lambda^-_t= \mu + \varphi \star dJ^+_s
\end{equation}
where $\varphi(t)$ is a causal (i.e., supported by $\mathbb{R^+}$), positive function and where
$\star$ stands for the convolution product $\varphi \star dJ_t = \int_{-\infty}^t \varphi(t-s) dJ(t)$.
The mean reversion property reads clearly from these last two equations :
the more $P_t$ goes up (resp. down), the greater the intensity $\lambda^-_t$ (resp.  $\lambda^+_t$) will be.
A criteria for the price increments and the intensities to be stationary  is given by $||\varphi||_1 <1$, where $||.||_1$ denotes the $\mathcal{L}^1(\mathbb{R})$ norm
(for a complete mathematical study of Hawkes process, see \cite{daley2003}).\\

\subsection{The Hawkes Impact Model (HIM) for market impact of a metaorder} 
We model the impact of a metaorder starting at time $t_0$, ending at time $t_0+T$ and corresponding to a continuous flow of buying\footnote{The impact of selling metaorder can be modeled using the exact same principles} orders with a trading rate $r_t$  supported by $[t_0,t_0+T]$ ($r_t\neq 0$ only for $t\notin[t_0,t_0+T]$) by a perturbation of the intensities. 

For the sake of simplicity, we will follow the microstructure model above and just consider mean-reversion reaction of the market (e.g., \cite{BM}, \cite{FZ}).
Let us point out that this is clearly not a realistic hypothesis if one is interested in mimicking precisely the microstructure. However, this is not our goal. In this section, we want to build a structural model that allows to explain the main dynamics of the market impact curve. 
In the same line as Bouchaud \cite{citeulike:3856893}, Gatheral \cite{doi:10.1080/14697680903373692} and \cite{BM}, 
we shall build a {\em linear} model, in the sense that the impact of the metaorder is nothing but the sum of the impact of its child order. 


\vskip .3cm
\noindent
{\bf The HIM model.} 
This model consists in replacing \eqref{hawkes_intensities_definition} by the two equations :
\begin{equation}
\label{hawkes_intensities_definition_continuous}
\lambda^+_t= \mu +  \varphi \star dJ^-_t + \int_{t_0}^tf(r_s)g^+(s-t_0)ds ~~{\mbox{and}}~~ \lambda^-_t= \mu + \varphi \star dJ^+_t + \int_{t_0}^tf(r_s)g^-(s-t_0)ds,
\end{equation}
where $f(r_s)ds$ (with $f(0) = 0$) codes the infinitesimal impact of a buy order of volume $r_sds$. The $f$ function corresponds to  the {\em instantaneous} impact function and $g^+$ and $g^-$ are the {\em impact kernel} functions. As empirical found in \cite{citeulike:1204462} and used by others authors before us (\cite{doi:10.1080/14697680903373692, citeulike:3856893}), we suppose the market impact can be separated in a factorized form: one depending on volume (or volume per time) and the other depending only on time.
\vskip .3cm
\noindent
{\bf The impulsive-HIM model : a particular choice for the kernels.} Following  \cite{BM}, it is reasonable to consider that the only "upward" impact of a single buying order is instantaneous, i.e.,  either the corresponding order ate up the whole first limit (in which case there is an instantaneous jump in the price) or it did not (in which case a limit order fills up the missing volume). This case corresponds to consider that $g^+$ is "purely" impulsive, i.e.,
\begin{equation}
g^+(t) = g^{+}_i(t) = \delta(t),
\end{equation}
where $\delta(t)$ stands for the Dirac distribution. 
 As for the "downward"component, we shall consider that  the market reacts to the newly arrived order as if it triggered an upward jump.
Doing so leads to the choice
\begin{equation}
\label{eq:constante_C}
g^{-}(t) = g^{-}_i(t) = {C} \frac{\varphi(t)}{||\varphi||_1},
\end{equation}
where $C>0$ is a very intuitive parameter that quantifies 
the ratio of contrarian reaction (i.e. impact decay)
and of the "herding" reaction (i.e. impact amplification). Indeed the $L^1$ norm of the herding reaction to an impulsive buying order 
is  $f(r_t)||g^+_i||_1 = f(r_t)||\delta||_1 = f(r_t)$ (with the notation $r_t$ for the instatanous trading rate), whereas the contrarian reaction to the same order is
$ f(r_t)||g^-_i||_1 = f(r_t)||\varphi||C/||\varphi|| = Cf(r_t)$.
Thus, one can distinguish 3 cases of interests (see \eqref{eq:permihim} of Proposition \ref{theprop} for analytical expressions) : 
\begin{itemize}
\item $C=0$ : no contrarian reaction; we expect a permanent effect on prices from metaorders,
\item $C=1$ : the contrarian reaction is as "strong" (in terms of the norm $||.||_1$) as the herding one. So we expect the two to compensate asymptotically, i.e., we expect the permanent effect of the metaorder on prices to be 0 (see Eq. \eqref{eq:permihim} of Proposition \ref{theprop}  for confirmation), 
\item $C\in]0,1[$ : the contrarian reaction is not zero but strictly smaller than the herding reaction.
\end{itemize}

Thus the impulsive-HIM model corresponds to the equations
\begin{equation}
\label{IHIM}
\lambda^+_t= \mu + \varphi\star dJ^-_t + f(r_t)~~{\mbox{and}}~~\lambda^-_t= \mu + \varphi\star dJ^+_t + C \int_{t_0}^tf(r_s)\varphi(s-t_0)ds,
\end{equation}
where $C$ is a positive constant that controls the contrarian vs. herding reaction of the market.

\subsection{Market impact curve within HIM}
\label{sec:HIM-res}
\noindent
According to our definition of market impact: \textit{the difference between the observed price moves and what it would have been without this specific order}, within HIM, the market impact of a metaorder (starting at time $t_0$) writes :
\begin{equation}
\label{theoretical_mi_definition}
\eta_t = \mathbb{E}[P_t],~~~\forall t\ge t_0.
\end{equation}

\noindent
Then, one can prove  (see Appendix \ref{proofs}) that :
\begin{proposition} ({\bf Transient, decay curves and permanent effect}) \\
\label{theprop}
In the framework of the HIM model \eqref{hawkes_intensities_definition_continuous}, for all $t\geq t_0$ ($t_0$ is the starting time of the metaorder), one has:
\begin{equation}
\label{eq:theoretical_mi_computation}
\eta_t =\int_{t_0}^{\infty} f(r_s)\Big(G(t-s) - (\kappa\star G)(t-s)\Big)ds,
\end{equation}
where
\begin{itemize}
\item $G(t) = \int_0^t(g^+(u)-g^-(u))du$
\item $\kappa = \sum_{n=1}^{\infty} (-1)^{n-1}\varphi^{(\star n)}$, where $\varphi^{(\star 1)}=\varphi$ and $\varphi^{(\star n)} = \varphi^{(\star n-1)} \star \varphi$.
\end{itemize}
In the case of the impulsive-HIM model \eqref{IHIM} , this formula gives
\begin{equation}
\label{res1}
\eta_t= \int_{t_0}^{t} f(r_s)H^C_{\varphi}(t-s)ds,\;\;t\geq t_0,
\end{equation}
where $H^C_{\varphi}(t) = 1-(1+C/||\varphi||_1)\int_{0}^t\kappa(s)ds$.
Moreover in the case of a constant rate strategy (i.e., $r_t=r,~ \forall t\in[t_0,t_0+T]$ and $r_t=0$ otherwise), the permanent effect of the metaorder on prices is 
\begin{equation}
\label{eq:permihim}
\eta_{\infty} = \lim_{t\rightarrow +\infty} \eta_t = f(r) T \frac {1-C}{1+||\varphi||_1}
\end{equation}
\end{proposition}
\noindent

\vskip .2cm
\noindent
Let us point out that several recent empirical results (\cite{BDM} and \cite{HBB}) seem to show that the Hawkes kernel $\varphi$ decays as a power-law.  Both studies found the exponent in the interval $[-1.5,-1]$.
The following corollary shows that, in the framework of the impulsive-HIM model, and in the case of a constant rate strategy, then  if $\varphi$ is power-law then the market impact curve asymptotically decays (to the limit permanent effect) as a power-law, with an exponent which is related to the exponent of $\varphi$. More precisely : 
\begin{cor}
\label{lem:linl_phi_psi}
In the framework of the impulsive-HIM model, let us consider a constant rate strategy, i.e., $r_t=r,~ \forall t\in[t_0,t_0+T]$ and $r_t=0$ otherwise. Assume that $\varphi$ is such that  
\begin{itemize}
\item $\varphi \ge \varphi^{(\star 2)}$ and,
\item $\exists K>0,~~~\lim_{t\to\infty}\varphi(t)t^{-b} = K$, with $b \in]-2,-1[$.
\end{itemize}
Then, the market impact curve decays to the permanent market impact $\eta_\infty$ asymptotically as a power-law with exponent $b + 1$, in the sense that 
\begin{equation}
\inf \left\{\gamma,~: \int_1^{\infty} (\eta_t-\eta_\infty) t^{-\gamma-1}dt <\infty\right\} = b + 1.
\end{equation}
\end{cor}
\noindent
The proof of this Corollary can be found in Appendix \ref{proofs}.

\vskip .2cm
\subsection{Qualitative understanding of the impulsive-HIM model.}
\label{marketmakers}
One can give a qualitative understanding of the impulsive-HIM model, especially of the meaning of $C$.
Assuming an idealized market where the metaorder is only traded against one market maker and potentially noise traders (i.e. no other metaorders), this impulsive-HIM model can be seen as modelling the market maker inventory $\calI=\lambda^+-\lambda^-$. 
This market maker accepts to provide liquidity to metaorders under his own risk limits. He tunes the level of interaction he accepts with metaorders (in short: the distance of his quotes to the mid-price and the inventory thresholds he uses to unwind his risk, stopping any interaction with the market for a while) using backtests or experience given his risk budget.
As a result: with a small risk budget he will provide very attractive quotes until his inventory crosses a given threshold.

It is typically the case for high frequency market makers (HFMM); they are confident their technology investment allow them to capture most of the flows on both sides of the book, and be able to detect fast that they are adversely selected. In this model they ``provide'' to the market a small value for the parameter $C$ of the impulsive-HIM model.
On the other hand, market makers with a large risk budget will provide lazier quotes and accept a larger inventory before unwinding it. It is typically the traditional behaviour of market makers in a quote driven market. In this model such market makers will provide a large value for the parameter $C$.

Thus, a market maker with a small $C$ does not provide a lot of resistance to the metaorder pressure and does not generate a large herding effect when he unwinds its inventory. A market maker with a large $C$  will generate large contrarian pressure to metaorders and, once his risk thresholds are crossed, will unwinding a large position, generating a large herding move (partly compensated by slower market makers with larger inventory).

Moreover, since market makers calibrate their $C$ to earn money under their risk constraint (and potentially other operational and regulatory constraints), each of them will ``specialize'' its activity around a given duration of metaorders $T(C)$. They will earn money providing liquidity to metaorders with a duration smaller than $T(C)$, and loose money interacting with longer metaorders.

Consequently, in the impulsive-HIM model, we have used a single market maker only with a given value of $C$ but in practice we should consider an extension of this model with a continuum of market makers, implementing a distribution of $C$ reflecting the distribution of metaorders in a given market.

This interpretation of the impulsive-HIM model can explain the concavity of transient impact stems from the market makers ``triggered'' by a given metaorder. Since $T(C)$ is increasing in $C$: at the start of the trading, the metaorder interacts with market makers with small $C$ (i.e. typically HFMM). After a while, such participants consider they are adversely selected, stop trading and unwind their inventory in front of market makers having a largest $C$. 
The longer the metaorder, the more market makers with large $C$, the more contrarian pressure and concavity to the transient impact. Once the metaorder stops, it is currently interacting with market makers with a perfectly adapted $C$ (from their viewpoint); the later now have to unwind slowly their inventory to realize their gain.

\vskip .2cm
\subsection{Back to real data}
Let us consider the case of a metaorder of fixed size $v$ and executed on a period of time $T$. In the framework of the impulsive-HIM model, for a constant trading rate $r_t = r = v/T$, Eq. \eqref{eq:theoretical_mi_computation} becomes:
\begin{equation}
\label{eq:mi_const_rate}
\eta_t = f(r)\left(1_{[t_0,t_0+T]}\star H^C_{\varphi}\right)(t).
\end{equation}
Let us point out that the trading rate $r$ has no influence on the shape of the transient market impact, it is just a multiplicative constant. In the following we want to use the impulsive-HIM model as a toy model to reproduce the transient and the market impact decay curves obtained in Sections \ref{sec:transient} and \ref{sec:decay} and displayed in Fig \ref{fig:subfigures}.

Following \cite{BDM} and \cite{HBB}), we choose to use a power-law kernel 
$$\varphi(t) = \varphi_{\alpha,b,\beta}(t) = \alpha(\gamma + t)^{b},$$ 
where $\gamma$ was fixed to $0.25$ (changing this value does not change drastically the following results).
Thus, apart from the instantaneous impact function $f$ (which is basically 
responsible for a rescaling of the whole market impact curve), there are three parameters left in the impulsive-HIM model, namely, $\alpha$, $b$, and $C$. The parameters $\alpha$ and $b$ are responsible for the endogenous mean-reverting activity of the market itself.
As pointed out in the previous section, the parameter $C$ encodes the proportion of contrarian effect to the impact.

\paragraph{Empirical measurements.}
Estimation is performed simultaneously on the four market impact curves $\{\hat \eta^{(i)}_s\}_{1\le i \le 4}$ (for $s \in [0,2]$) displayed on Fig.  \ref{fig:subfigures} (from (a) to (d)). The parameters are  $\alpha$ and $b$ (these two parameters are shared by all the curves) and the parameters  $\{C^{(i)}\}_{1\le i \le 4}$  corresponding to four types of market makers as explained in Section \ref{marketmakers}.
Estimation follows the three following principles
\begin{itemize}
\item We constrain the estimation to fit the temporary market impact $\hat \eta^{(i)}_{s=1}$ of each curve $i\in\{1,..,4\}$.
\item  Each curve $\hat \eta^{(i)}_{s}$ is as a function of the renormalized time $s$ whereas the model gives the market impact $\eta_t$ (Eq. \eqref{eq:mi_const_rate}) as a function of the physical time, rescaling in time must be performed independently on each curve.
The corresponding rescaling parameter has been chosen to be the largest duration in the corresponding time ranges, i.e., $T = T_1:=15$ min for curve in Fig. \ref{fig:subfigures}(a), $T=T_2:=30$ min for curve in Fig. \ref{fig:subfigures}(b),
$T=T_3:=60$ min for curve in Fig. \ref{fig:subfigures}(c) and $T=T_4:=90$ min for curve in Fig. \ref{fig:subfigures}(d).
\item To account for the instantaneous impact function $f(.)$, each curve is rescaled independently.
\end{itemize}
Thus, the estimation procedure sums up in 
$$
(\hat{\alpha}, \hat{b}, \hat{C_1}, \hat{C_2}, \hat{C_3}, \hat{C_4}) = \argmin_{(\alpha, b, \hat{C_1}, \hat{C_2}, \hat{C_3}, \hat{C_4})} %
 \sum_{i=1}^{4}\int_0^2  \left(  \frac{\hat{\eta}_1^{(i)}}{\eta_{T_i}} \eta_{sT_i}- \hat{\eta}_s^{(i)}     \right)^2 ds
$$
or (in detail to underline the dependencies in the parameters, with $\bone_{[T]}:=1_{[t_0,t_0+T_i]}$):
$$\def\HCi{\left(\bone_{[T(\omega)]}\star H_{\phi_{\alpha,b}}^{C_i}\right)}
\argmin_{(\alpha, b, C_1,C_2,C_3,C_4)} %
\sum_{i=1}^{4} \int_0^2 \left(  \frac{\hat{\eta}_1^{(i)}}{\HCi(T_i)} \HCi(sT_i)%
- \hat{\eta}_s^{(i)}     \right)^2 ds
$$

The value we find for $\hat{b}$ is close to $-1.5$ and for $\alpha$ we find a value such that $||\varphi||_1\approx 0.8456$
which is rather close to the critical value 1. These results are in good agreement with the works  \cite{BDM} and \cite{HBB}.
For the parameters  $\{C_i\}_{1\le i\le 4}$, we find $\hat{C}_1\approx 0.5$, $\hat{C}_2\approx 0.70$, $\hat{C}_3\approx 0.80$ and $\hat{C}_4\approx 0.85$.

These results correspond to the qualitative interpretation of the impulsive-HIM model: $C$ \emph{increases with} $T$. It could reflet the mix of market maker implicitly selected by the metaorder given its duration.


The so-obtained fits are shown in Fig. \ref{fig:subfigures}. 
One can see that the model impulsive-HIM is able to reproduce precisely the  shapes of both transient and market impact decay curves. Moreover, it also reproduces the  dependence on $T$ of the curvature of the transient market impact, i.e., the smaller $T$, the more linear the transient market impact.

\section{Investment strategies, price anticipation and permanent impact}
\label{sec:permanent}

\subsection{Positioning}

The permanent impact, as stated in the introduction, is the remaining price shift after the decay has taken place. So far very few papers, even if it has been extensively studied, addressed the problem of the permanent impact at the daily scale. Among the papers which studied permanent impact one distinguishes two different positions. The first position, which is shared by number of econophysists, describes the permanent impact as the consequence of a mechanical process. The second position, which is further favoured by economists, considers the permanent impact as the trace of new information in the price.
\begin{itemize}
\item  In a purely mechanical vision stock prices move because of the trading of all the market participants. If the sell pressure is stronger than buying pressure the price goes down and vice versa if the buying pressure is the strongest. The purpose of the econophysicist is to understand the behaviour of these two forces and to find the law at the root of the impact of buying/selling pressure on prices' dynamic. See \cite{citeulike:7360166}{}, \cite{bouchaud} and \cite{citeulike:13266538}.
\item On the other side, the informational vision says stock prices move because new information is made available to market participants. According to this new information investors update their expectations changing their offer and demand which leads to a new global equilibrium resulting in new price levels. In a seminal paper, \cite{citeulike:3320208}, the author shows how prices are driven to their new level through the execution of a metaorder by an informed trader. The informed trader is constantly adjusting her trading to her observation of the price in real time: she increases or reduces pressure whether the price is too far or not from the targeted price.
\end{itemize}
Between those two ``pure'' visions comes a range of papers analysing all sorts of metaorder databases and reaching conclusions which tilt the scale in favour of one position or the other. Ourselves in the present paper we are not going to settle for one or the other vision. We believe the mechanical - informational dual vision about the nature of permanent effect renders a good picture of the phenomenon. As every dual paradigm, this is a well known principle by physicists, taken separately the two pure concepts fail to give a satisfying picture of the whole phenomenon.

\paragraph{Focus on permanent market impact.}
The mechanical vision of the permanent effect is not in one single piece. 
Not all the ``\emph{mechanists}'' believe permanent impact exists:
\begin{itemize}
\item In the picture of \cite{FGLW} and of \cite{BR} permanent market impact is important and roughly equals to $2/3$ of the temporary impact on a metaorder by metaorder basis. This is the so-called fair pricing hypothesis.
\item In the picture of \cite{bouchaud}, there is no such thing as permanent impact. The author argues that what is called permanent impact is in fact the result of the long memory of the sign of the metaorder flow. This picture is incompatible with the permanent impact hypothesis because long memory of order flow would result in trending stock prices and thus contradicting the market efficiency principle.
\end{itemize}

\paragraph{Our methodology.}
Recently \cite{citeulike:13266538} exposed a new methodology to remove the informational content of the studied proprietary metaorder database. Each metaorder is characterized by the intensity of the trading signal which triggered it. 
The divergence between the effective and the expected price moves gives a measure of the permanent impact in the absence of informational effect. Based on that methodology the authors observed no permanent impact on the residual price moves: after a dozen of days 100\% of the impact seems to have completely vanished.

We will study our brokerage metaorder database at the daily scale with the same angle as \cite{citeulike:12825932}. In this important contribution to the subject of the permanent impact, the authors studied the decay of the market impact until several days after the execution of a trade. 
In their specific case, the studied database allows to tag the decision at the origin of each trade.
It differentiates metaorders which have been triggered by an informational change that could impact the price from the others. More specifically the authors divided their metaorder database into two groups:
\begin{itemize}
\item Those stemmed from redemptions or new subscriptions, thus triggered by heterogeneous and relatively exogenous to the market information. Those metaorders are called ``uninformed trades'' or ``cash trades''.
\item The rest, essentially metaorders coming from portfolio rebalancing, is designated the set of ``informed trades''.
\end{itemize}
\cite{citeulike:12825932}{} shows that on a daily basis, as of the execution day until 60 days after, ``informed trades'' have permanent impact but ``uninformed trades'' do not.

To go one step further using our database, we will try to remove as much informational effect as possible from the price move. Since we use the database of an executing broker, we do not know explicitly what triggered the creation of the metaorders, contrary to authors of \cite{citeulike:13266538} who had access to the prior belief of the metaorder issuer on future price changes. Still we can assume that the clients of the broker globally constitute a good representative sample of the diversity of market investors. Thus the total portfolio maid up of all the metaorders is not far from the CAPM market portfolio that is to say the beta of this portfolio equals to 1. Therefore we studied the price moves on the post-execution period net of the market portfolio moves, we then looked at the so-called idiosyncratic moves of stock prices.

Our results are compatible with those obtained in \cite{citeulike:13266538}, where after removing the alpha from price moves no permanent impact left, or in \cite{citeulike:12825932}, wher ``uninformed trades'' add no permanent impact: we found no permanent impact on the idiosyncratic component of the prices.


\subsection{Debiasing temporary impact of metaorders traded during the post-execution period}
When studying the average profile of the post-execution price moves particular attention should be paid to the presence of autocorrelations in the metaorder flow. Indeed the  temporary impact of correlated metaorders executed after each studied one cannot be considered as the permanent impact of the metaorder executed on the execution day. In \cite{FarmerLillo2006PriceImpact} the authors discussed the issue of price efficiency and of order flow correlations. \cite{citeulike:12825932} recognized the presence of autocorrelations in the flow of metaorders and used a simple market impact model to withdraw the impact of the metaorders traded over the post-execution period. In \cite{citeulike:13266538}, authors used a more sophisticated version of such a cleaning, deconvoluting the decay ``kernel'' from the data; in this paper we perform the same debiasing precedure as in \cite{citeulike:12825932}. Figure~\ref{fig:mk_perm:autocorr} shows the autocorrelations of the market participation rate for the metaorders of our database (bootstrapped quartiles and median). One clearly sees that the autocorrelation persists beyond 20 days after the execution day.

Removing the effect of the autocorrelations of the metaorder flow in order to get an unbiased long term impact picture is not straight forward. The temporary impact of metaorders on the price moves turns out to be far from linear according to the daily participation rate. Following \cite{citeulike:12825932} we fitted a square root model\footnote{Section \ref{sec:temporary} of this paper confirms a power law close to a square root. We use a square root to avoid overfitting and to obtain more robustness. We checked that a power of 0.4 to 0.7 gives similar results.} of the daily participation rate on the temporary impact. To debiased price moves of the executed metaorders on the post-execution period we simply subtracted the temporary impact associated to the aggregated daily participation rate on that date applying our square root model.

\begin{figure}[ht!]
\begin{center}
	\includegraphics[width=1\linewidth]{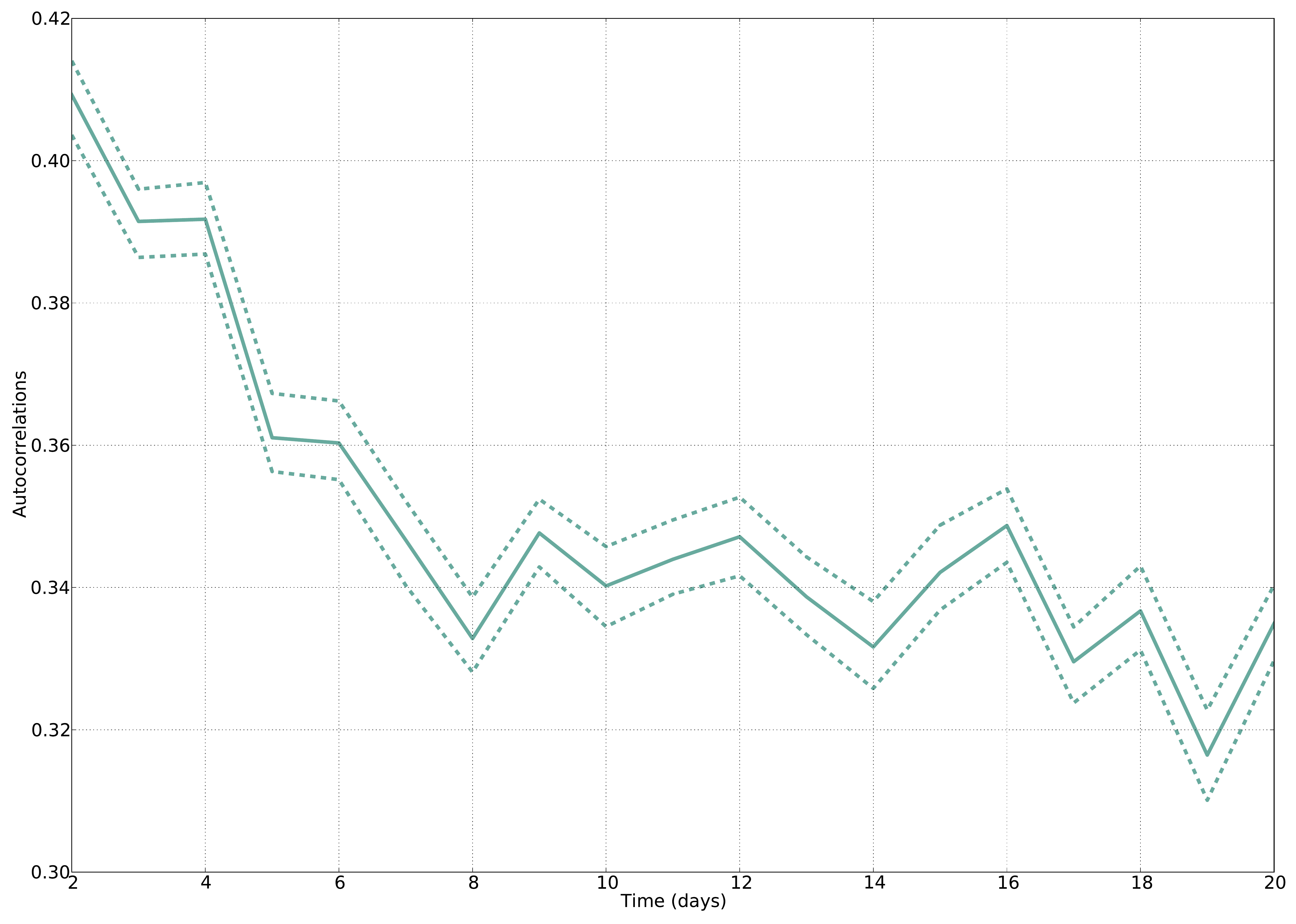}
        \caption{Autocorrelations of the market participation of metaorders with lags in $\{1,\ldots,20\}$,. The dashed curves represent the first and the last boostrapped quartiles, the solid curve is the bootstrapped median.}
	\label{fig:mk_perm:autocorr}
\end{center}
\end{figure}

The meaning of such a ``debiasing'' step needs to be commented. Assume the following initial event: at some point in the time an information potentially affecting the value of a stock is delivered to investors. This news is available at a precise moment, but the reaction of investors to it is diffuse. They will not simultaneously take it into account in their portfolios. Each of them have constraints and processes different enough so that their decisions will span over several days. Each time an investor trades according to this information, she impacts mechanically the price, eventually making the impact of news on prices diffusive. To reconstruct the immediate impact we gathered all the metaorders on date $d=1$ and averaged price moves net of the impact of the metaorders executed the next days by subtracting the square root model of the daily participation rate of those metaorders.


\paragraph{Price moves breakdown into systematic and idiosyncratic components.}
For each metaorder $\omega$ we define a CAPM-like decomposition into a systematic and an idiosyncratic component centred on the execution day over a 41 days period (20 days before and 20 days after the studied day $D$),
\begin{equation}\label{eq:mk_perm:CAPM}
\forall d \in\{D-20,\ldots,D,\ldots,D+20\},\quad\log(P_d)-\log(P_{d-1}) = %
	\beta(\omega)(\log(\Ind_d)-\log(\Ind_{d-1}))+\Delta\IdioComp_d,
\end{equation}
where $D$ is the date of metaorder execution, $P_d$ is the stock's close price on date $d$, and $\beta(\omega)$ implicitly designates the beta of the traded stock on the period from $D-20$ to $D+20$.

We use the reference index for each stock, which price is noted $\calI_d$ on day $d$.
\begin{itemize}
\item We define the idiosyncratic component as the cumulative sum of the $\Delta \IdioComp_d$ from (\ref{eq:mk_perm:CAPM}): $\sum_{k=D}^d \Delta(\IdioComp_k)=\IdioComp_d - \IdioComp_{D-1}$;
\item Similarly we define $\beta(\log(\Ind_d)-\log(\Ind_{D-1}))$ as the systematic component.
\end{itemize}
The systematic, the idiosyncratic components and their sum are considered relatively to the close prices one day before the execution, $D-1$. This cumulative version of (\ref{eq:mk_perm:CAPM}) is also called the post-execution profile and can be considered as an estimate of the permanent impact.

\subsection{Result analysis and figures}
The two figures \ref{fig:mk_perm:post_exec_mi} and \ref{fig:mk_perm:post_exec_unbiaised_mi} present the idiosyncratic, the systematic and the total daily post-execution profiles of price moves. Figure \ref{fig:mk_perm:post_exec_mi} shows ``not-yet-debiased'' post-execution profiles. The observed price jump between day 0 and day 1 is the daily market impact. This jump is visible on the idiosyncratic and the systematic components. Over the period extended from day 1 until day 20 the prices are slowly trending in the same direction as the jump on the execution day and no reversal is seen at all. This holds true for idiosyncratic and systematic components as well.

Figure \ref{fig:mk_perm:post_exec_unbiaised_mi} shows the debiased post-execution profile. The market impact of the metaorders executed the day after the execution day has been removed. Over the post-execution period the price is converging back to a level lower than the one reached on execution day. The idiosyncratic post-execution profile is even reaching its initial level before the end of the observation period of 20 days. Thus the remaining permanent impact 20 days after the execution is entirely explained be the systematic component that is to say by the average level of the market. Going back to our prior about the nature of the metaorder database: the global portfolio constituted of all the traded stocks is the market portfolio; we can say that:
\begin{itemize}
\item once the temporary impact of the correlated metaorders the days after have been removed (see figure \ref{fig:mk_perm:autocorr}),
\item and once the aggregated information (the systematic component) potentially used by metaorders issuers has been isolated and removed,
\end{itemize}
there is \emph{no remaining permanent effect due to the metaorder itself}.
\begin{figure}[ht!]
\begin{center}
	\includegraphics[width=1\linewidth]{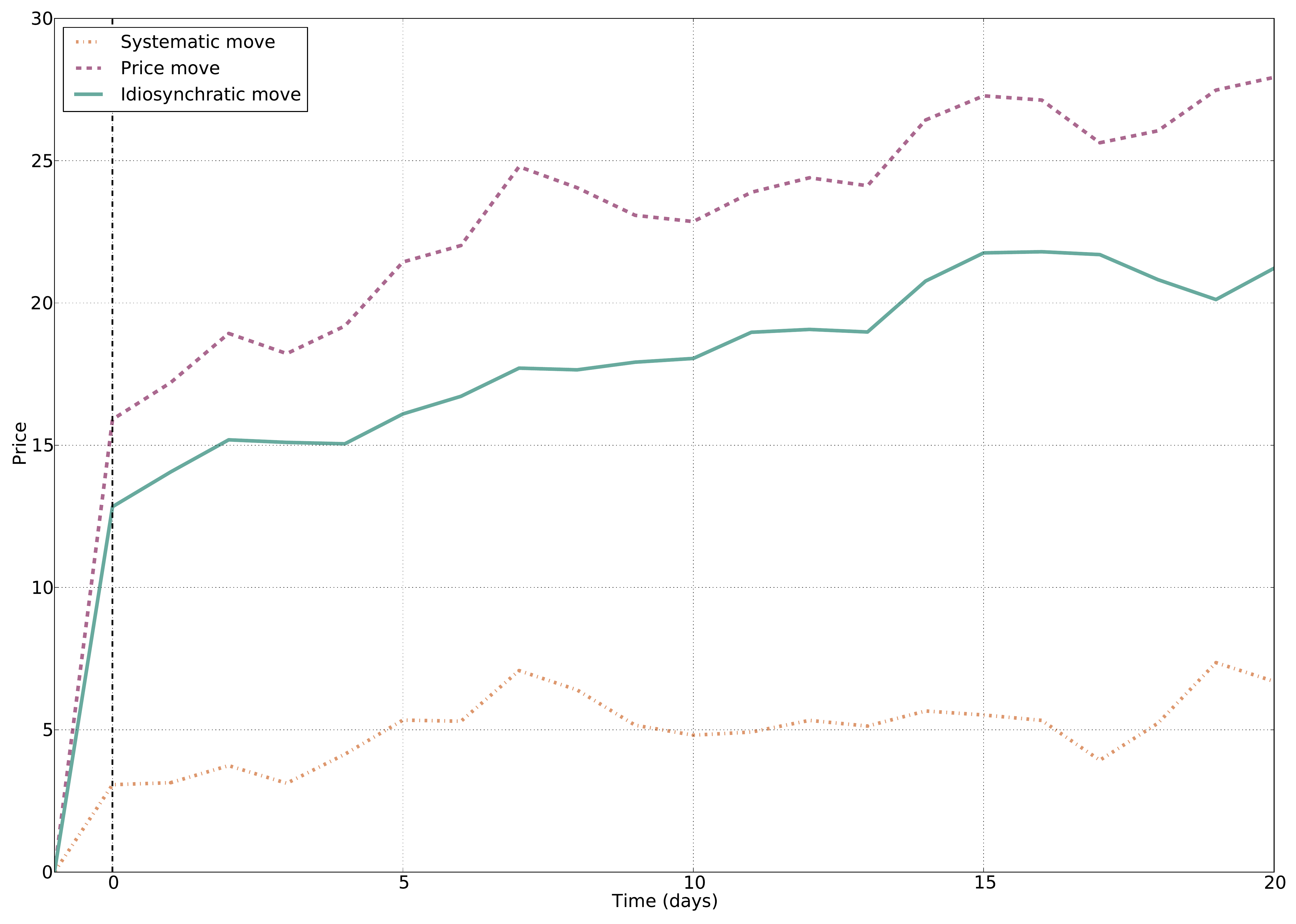}
        \caption{Post-execution profile relatively to close price on the day before execution (units = basis points). Idiosyncratic component + Systematic component = Total component.}
	\label{fig:mk_perm:post_exec_mi}
\end{center}
\end{figure}

\begin{figure}[ht!]
\begin{center}
	\includegraphics[width=1\linewidth]{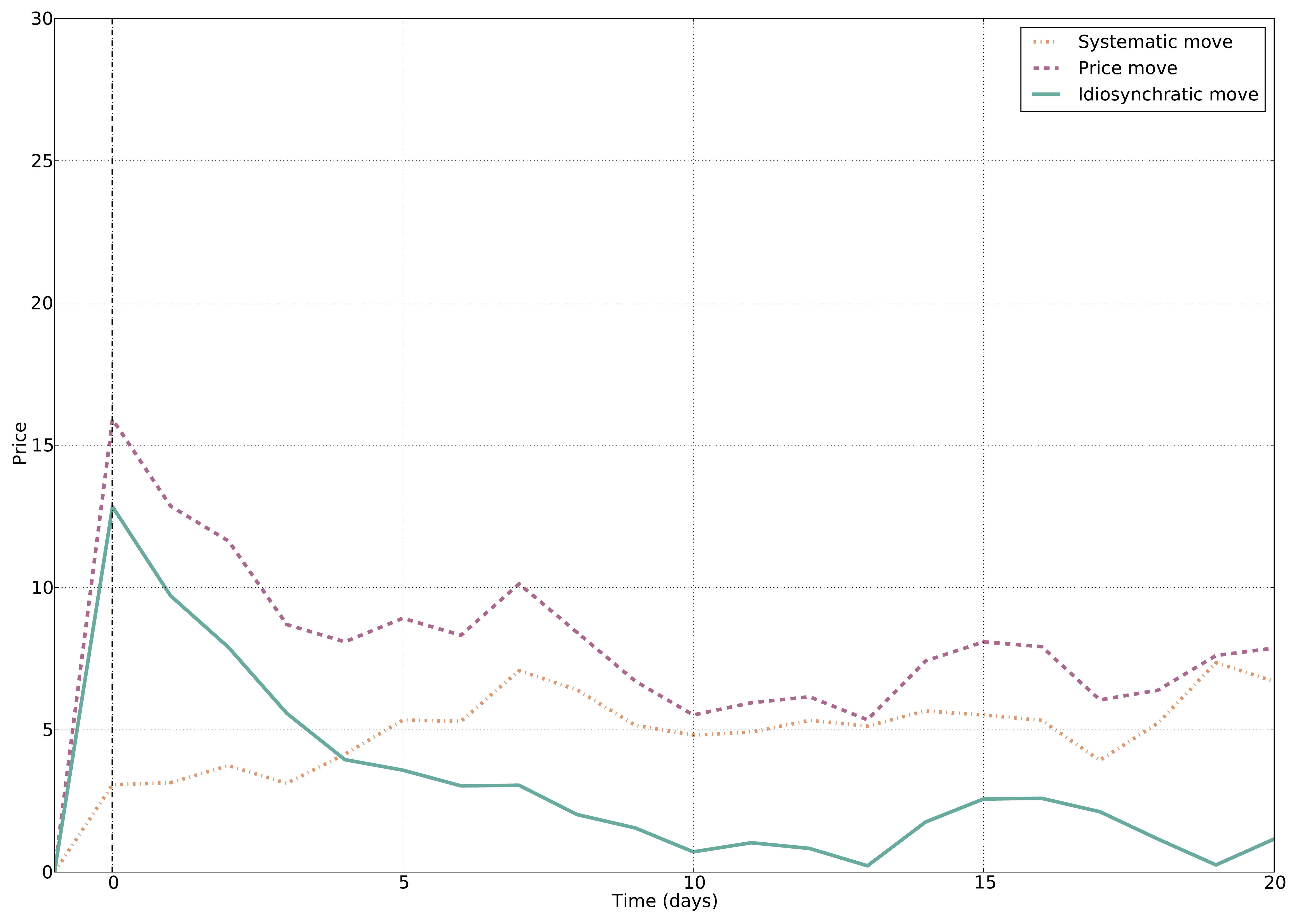}
	\captionof{figure}{Post-execution profile without the impact of other metaorders. Price moves are considered relatively to close price the day before execution (units = basis points). Idiosyncratic component + Systematic component = Total component.}
	\label{fig:mk_perm:post_exec_unbiaised_mi}
\end{center}
\end{figure}


%
\section{Synthesis: Contribution of the life cycle of an large order to the price formation process}

In this paper, we studied a database of metaorders issued by a large execution broker in Europe over 2010.
The specificity of this database is that: it is concentrated in time (it spans one year), and geographically (one zone only), and all the metaorders have been traded electronically using trading algorithms with a very stable trading rate.

We studied the market impact of these metaorders at any scales: intraday, splitting the impact in three phases (transient, temporary, and decay), and daily, with a focus on the long term decay and potential remaning permanent effect of the metaorders in the prices.

Moreover, we proposed a toy model based on Hawkes processes (impulse HIM) to illustrate qualitatively our empirical results and guide our explorations.

For our intraday study, we provided some methodological elements, to make the difference between the use of log-log regression, and direct regression with different metrics. It stems temporary market impact scales in power law of the daily participation, with a smaller exponent when we use robust estimation approaches rather than rough log-log regression.
Thanks to our direct regressions, we identified a duration effect on temporary impact, than may be read at the light of the specificity of our database. The effect of different duration is easier to identify on trading algorithms using an almost constant trading rate than on ones using rates free to change to adapt to price or liquidity opportunity signals.

The transient impact is power law in time according to a power around 0.6, compatible with existing papers. We see there does not seem to be any ``anticipation by the market'' of the duration of the metaorder from the transient phase viewpoint. 

We fit power laws on the decay of the impact, and find power around 0.6, confirming previous academic papers. We cannot exclude a change of slope (or power) after some time, that could lead to a decay in two phases.
In any case the price move does not come back to normal on average before the end of the day.
Our impulse HIM model predicts asymptotically a decay in power laws, and exhibit a crucial parameter $C$, which can be used to characterizes the ratio between contrarian and herding activity in reaction to a metaorder.

In the intraday part of the study, we had not the capability to split the price moves between an informational and a mechanical component. The informational component being the dependence between the fact the metaorder issuer took the investment decision and the future price moves: good portfolio managers should see the price going up during (and after) a buy, and the price going down given they decided to sell. The mechanical component being the ``remaining'': how the very metaorder liquidity consuming pressure moves the price by itself.

In the daily part of the study, we assumed most of the clients of the executing broker were institutional investors anticipating betas in the sense of the CAPM (i.e. market moves). We thus removed the beta part of the price move before studying market impact. After a cleaning phase dedicated to take care of the positive autocorrelation of metaorders as observed in the database, we show the remaining permanent effect goes to zero in a dozen of days, in line with existing studies on cash trades or on metaorders associated with known price anticipation.

This paper focus the attention on methodological aspects and on the importance of taking into account the informational content of metaorders before studying their mechanical impact. This mechanical impact being the main component of trading costs of institutional investors, it is of paramount important to define it properly and to understand it. Our study, confirming and clarifying the few existing ones, opens two doors.
On the one hand, summarizing and unifying existing frameworks for daily and intraday data mining of impact, can be used by other researchers to understand the variations of the parameters of this important component of trading costs, shedding light on the influence of market microstructure on it. Provided than datasets would be available, it could be of great help for regulators. 
On the other hand, our Hawkes based toy model (impulse HIM) could be used for more theoretical work, to obtain a better understanding of the interactions between liquidity providers and liquidity consumers on markets.

\paragraph{Acknowledgements.}
Authors would like to thank 
Jonathan Donier,
Marc Hoffmann, 
Julianus Kockelkoren, 
Iacopo Mastromatteo 
and 
Bence Toth
for fruitful discussions that helped to improve this paper.
We gratefully acknowledge financial
support of the chair {\it Financial Risks}
of the {\it Risk Foundation},
of the chair
{\it Mutating Markets} of the French Federation of Banks and
of the chair
QuantValley/Risk Foundation: Quantitative Management Initiative.



\bibliographystyle{apalike}
\bibliography{MI}


\clearpage
\appendix

\section{Proofs}
\label{proofs}
\textbf{Proof of Proposition \ref{theprop}}

\noindent
This is a consequence of a more general result on Hawkes Processes:
\begin{theorem}
Let $(N_1,\dots,N_d)$ be a $d$-multivariate Hawkes process defined by his intensity $\lambda=(\lambda_1,\dots,\lambda_d)$:
\begin{equation}
\label{eq:def_Hawkes}
\lambda_t = \mu(t) + \int_{[0,t)} \Phi(t-s)dN_s
\end{equation}
where $\mu(t) = (\mu_1(t),\dots,\mu_d(t))$ is a vector of functions 
from $\mathbb{R}_+$ to $\mathbb{R}_+$ and
$\Phi(t)=(\varphi_{ij}(t))_{1\leq i,j \leq d}$ is a matrix containing functions from $\mathbb{R}_+$ to $\mathbb{R}_+$. \\
Under the assumptions:
\begin{itemize}
\item $\int_0^t \mu(s)ds < \infty\;\;\;\forall t>0,$
\item  $\textrm{the spectral radius of the matrix\;} K=\int_0^{\infty} \Phi(t)dt \textrm{\;is inferior to\;} 1,\textrm{\;}\rho(K)<1,$
\end{itemize}
we have
\begin{equation}
\label{eq:mean_hawkes}
\mathbb{E}[N_t] = h(s) + \int_0^t \Psi(t-s)h(s)ds,
\end{equation}
where $\Psi = \sum_{n\geq 1} \Phi^{(\star n)}$ and $h(s) = \int_0^t \mu(s)ds$,
where $\Phi^{(\star n)} = \Phi \star \ldots \star \Phi$ (with $n$ terms on the right hand-side) and where the convolution product of two matrices $A(t)=\{a_{ij}(t)\}$ and $B(t)=\{b_{ij}(t)\}$ is defined as the matrix $C(t)=\{c_{ij}(t)\}$, such that
$$
c_{ij}(t) = \sum_k a_{ik} \star b_{kj}.
$$
\end{theorem}
\begin{proof}
Let us first remark the elements of the matrix $\Psi$ are in $L^1$. Indeed, by induction we have,
$\int_0^{\infty} \Phi^{(\star n)}(t)dt = K^n$ and  since $\rho(K)<1$, the series $\psi=\sum_{n \geq 1} K^n$ is finite component-wise, thus one gets:
\begin{equation}
\label{eq:int_Psi}
\int_0^{\infty} \Psi(t)dt = K(Id-K)^{-1}.
\end{equation}
We will now show that:
\begin{equation}
\mathbb{E}[N_t] = \int_0^t \mu(s)ds + \mathbb{E}[\int_0^t \Phi(t-s)N_sds]\quad \forall t>0
\end{equation}
Using that $N_t - \int_0^{t} \lambda_sds$ is a $(\mathscr{F}_t)$-martingale (where $(\mathscr{F}_t)_{t\geq 0}$ is the $\sigma$-algebra generated by the random variables $N^i_s; s\leq t; 1\leq i \leq d$), we have:
$$\mathbb{E}[N_t] = \mathbb{E}[\int_0^{t} \lambda_sds] = \mathbb{E}[\int_0^{t}\mu(s)ds] + \mathbb{E}[\int_0^{t}ds\int_0^s \Phi(s-u)dN_u].$$
But, by Fubini's theorem:
$$\int_0^{t}ds\int_0^s \Phi(s-u)dN_udu = \int_0^{t}\Big(\int_s^{t}\Phi(t-u)dt\Big)dN_u = \int_0^{t} \Big(\int_0^{t-s}\Phi(s)ds\Big)dN_u.$$
We denote $F(t) = \int_0^t \Phi(s)ds$ and using an integration by parts we have:
$$\int_0^{t} F(t-s)dN_s = \Big[ F(t-s) N_s\Big]_0^{t} + \int_0^{t} \Phi(t-s)N_sds=$$
$$=F(0)N_{t}-F(t)N_0 + \int_0^{t} \Phi(t-s)N_sds = \int_0^{t} \Phi(t-s)N_sds.$$
So we obtain:
$$\mathbb{E}[N_{t}]=\mathbb{E}[\int_0^{t}\mu(s)ds] + \mathbb{E}[\int_0^{t} \Phi(t-s)N_sds].$$
Using once again Fubini's theorem:
\begin{equation}
\mathbb{E}[N_t]= \int_0^t \mu(s)ds + \int_0^t \Phi(t-s)\mathbb{E}[N_s]
\end{equation}

This is a classical renewal equation and the solution is given by (\ref{eq:mean_hawkes}). The interested reader can find more on renewal theory on the book of David Cox (\cite{Cox}).\\
\end{proof}

Let us now prove the first part of Proposition \ref{theprop}. In order to ease further
notation we take $t_0=0$. We apply the previous theorem in the particular case of the $2$-dimensional Hawkes process $(J^+,J^-)$ with $\Phi=\begin{pmatrix}
   0 & \varphi \\
   \varphi & 0
\end{pmatrix}$
and we successively compute:
\begin{itemize}
\item $\Phi^{\star n}=\begin{pmatrix}
   0 & \varphi^{\star n} \\
   \varphi^{\star n} & 0
\end{pmatrix}$ if $n$ is even and $\Phi^{\star n}=\begin{pmatrix}
   \varphi^{\star n} & 0 \\
   0 & \varphi^{\star n}
\end{pmatrix}$ if $n$ is odd.
\item $h(t)=(h_1(t),h_2(t))= (t\mu+\int_0^t (f(r)\star g^+)(u)du,t\mu+\int_0^t (f(r)\star g^-)(u)du)$
\item $\eta_t=\mathbb{E}[J^+_t-J^-_t] = (h_1(t)-h_2(t)) - \kappa \star (h_1-h_2)(t)$.
\end{itemize}
This proves Eq. \ref{eq:theoretical_mi_computation} since $h_1-h_2=f(r)\star G$.

\vskip .3cm
\noindent
In the case of impulsive-HIM, we set
$H(t) = G(t)-G\star \kappa(t)$ and $C_1 = C/||\varphi||_1$. Let us compute the derivative of $H$ for $t > 0$ : 
$$H'(t)=G'(t)-G'\star \kappa (t)= -C_1\varphi(t)-(\delta-C_1\varphi)\star \kappa(t) =  -C_1\varphi(t) -\kappa(t)
+C_1 \varphi \star \kappa(t)$$
Using:
$$\varphi\star\kappa = \varphi\star \sum_{n=1}^{\infty} (-1)^{n+1}\varphi^{(\star n)} = \sum_{n=2}^{\infty} (-1)^{n}\varphi^{(\star n)} = \varphi - \kappa,$$
we obtain
$$H'_t = -(1+C_1)\kappa(t).$$
This allows to find  $H(t) = 1 -(1+C_1) \int_0^t \kappa(s)ds$ which proves \eqref{res1}. Eq. \eqref{eq:permihim} 
is a direct consequence of this equation when the rate $r_t$ is constant. 

\vskip .5cm
\noindent
\textbf{Proof of Corollary \ref{lem:linl_phi_psi}} \\
This Corollary is a direct consequence of the following Lemma (which links the power-law exponent of $\varphi$ with the one of $\kappa$) and of the expression \eqref{res1}.
\begin{lem}
Let $p \in [0,1]$. Under the assumption $\varphi \geq \varphi^{(\star 2)}$,  $\int_{[0,\infty)} t^p\varphi(t)dt < \infty$ if and only if $\int_{[0,\infty)} t^p\kappa(t)dt < \infty$.
\end{lem}
\begin{proof}
We recall that $\kappa =\sum_{n=1}^{\infty} (-1)^{n+1}\varphi^{(\star n)}$. Which in the Fourier domain becomes:
$$\hat{\kappa}(\omega) =\sum_{n=1}^{\infty} (-1)^{n+1}\hat{\varphi}^n(\omega)=\frac{\hat{\varphi}(\omega)}{1+\hat{\varphi}(\omega)},$$
By inverting the previous formula we get:
$$\hat{\varphi}(\omega) = \frac{\hat{\kappa}(\omega)}{1 - \hat{\kappa}(\omega)} = \sum_{n=1}^{\infty} \hat{\kappa}^n(\omega).$$
Returning in the time-domain we have:
\begin{equation}
\varphi = \sum_{n=1}^{\infty} \kappa^{(\star n)}.
\end{equation} 
Let us remark the asumption $\varphi \geq \varphi^{(\star 2)}$ ensures that $\kappa$ is positive.  
We are ready for the proof
\begin{itemize}
\item Let us assume that $\int_0^{\infty} t^p \varphi(t)dt < \infty$. \\
Since $\kappa(t)\geq 0$, it is easy to see that $\kappa^{(\star n)}\geq 0$. So $0 \leq \kappa(t) \leq \varphi(t)$ and than $\int_0^{\infty} t^p \psi_d(t)dt < \infty$ is evident.
\item Let us assume that $\int_0^{\infty} t^p \kappa(t)dt < \infty$.\\
Let $I_n = \int_0^{\infty} t^p \kappa^{(\star n)}(t)dt$, $\forall n\geq 1$, and $\int_0^{\infty} \kappa(t)dt = c$, $c \in \mathbb{R}_+$. Using that the function $t^p$ is concave for $p\in[0,1]$, we get:
$$I_{n+1} = \int_0^{\infty} t^p \Big(\int_0^t \kappa(t-s)\kappa^{(\star n)}(s)ds\Big)dt = \int_0^{\infty}  \Big(\int_0^{\infty} (t+s)^p \kappa(t)dt\Big)\kappa^{(\star n)}(s)ds$$
$$\leq \int_0^{\infty}  \Big(\int_0^{\infty} (t^p+s^p) \kappa(t)dt\Big)\kappa^{\star(n)}(s)ds  = \int_0^{\infty}\Big( I_1 + cs^p \Big)\kappa^{(\star n)}(s)ds  = c^nI_1 + cI_n.$$
Therefore for all integer $N$:
$$\sum_{i=1}^N I_n \leq I_1 + \Big( \sum_{n=1}^{N-1} c^n \Big)I_1 + c \sum_{n=1}^{N-1} I_n.$$
And we easily obtain:
$$\sum_{n=1}^{N-1} \leq \frac{I_1}{(1-c)^2}$$
Thus, for $N \to \infty$:
$$\int_0^{\infty} t^p \varphi(t)dt = \sum_{n=1}^{\infty} I_n \leq \frac{I_1}{(1-c)^2} < \infty$$
\end{itemize}
\end{proof}

\section{Trading algorithms}
\label{sec:algonames}


\begin{itemize}
\item a PoV (i.e. Percentage of Volume) is a trading algorithm for which the volumes of the transactions stays in a narrow band (of width of the order of a few percents) around a constant chosen as a fixed fraction of the estimated daily volume. Thus, "smart" limit orders are sent as long as the integrated volume is in the band. When it is no longer the case, these limit orders are canceled and market orders are sent instead.
\item a VWAP (i.e. Volume Weighted Average Price) is a trading algorithm parameterized by a start time and an end time, which tries to make the integrated transaction volume  to be as close as possible to
the \emph{average intraday volume curve} of the traded security (i.e. the \emph{U-shaped pattern} on the US stocks for instance, see \cite{citeulike:12047995} Chapter 2.1. for details about fixing auctions and intraday volume curves around the world).
It means that the transaction volumes  are higher (resp. lower) during the period of high (resp. low) averaged activity.
\item a IS (Implementation Shortfall) is the typical implementation of
  an Almgren-Chriss like algorithm \cite{citeulike:10673860}. It targets to minimize the expected slippage between the average price and the decision price, including a penalization by a risk factor exposure (and using an instantaneous market impact model, usually in square root).
\end{itemize}
For more details about trading algorithms, see Chapter 3.3 of \cite{citeulike:12047995}.

\section{Databases}
\label{sec:datatables}

\begin{table}[h!]
  \centering
  \begin{tabular}{|c|c|c|c|c|} \hline
    Stock Exchange & $\Omega^{(day)}$ & $\Omega^{(te)}$ & $\Omega^{(tr)}$ & $\Omega^{(de)}$ \\\hline
Amsterdam & 4.89$\%$ & 5.43$\%$ & 5.18$\%$ & 5.65$\%$  \\\hline
Frankfurt  & 13.34$\%$ & 13.27$\%$ & 13.37$\%$ & 14.54$\%$  \\\hline
London  & 25.84$\%$ & 25.56$\%$ & 24.13$\%$ & 21.62$\%$  \\\hline
Madrid  & 4.71$\%$ & 5.02$\%$ & 5.05$\%$ & 5.21$\%$  \\\hline
Milan &  6.08$\%$ & 5.98$\%$ & 5.34$\%$ & 5.120$\%$  \\\hline
Paris & 22.64$\%$ & 26.04$\%$ & 27.94$\%$ & 29.90$\%$  \\\hline
Others & 22.50$\%$ & 18.67$\%$ & 18.95$\%$ & 17.92$\%$  \\\hline
  \end{tabular}
  \caption{Meta-orders distribution on different European Stock Exchanges for different databases}
  \label{tab:data:dist}
\end{table}

\begin{table}[!h]
\centering
	\begin{tabular}{|c|c|c|c|c|c|c|c|}\hline
Characteristic & Database & Mean & Q5\% & Q25\% & Q50\% & Q75\% & Q95\% \\\hline
\multirow{4}{*}{$N$} & $\Omega^{(day)}$ & 73 & 11 & 19 & 50 & 80 & 249\\
& $\Omega^{(te)}$ & 77 & 11 & 20 & 40 & 87 & 261\\
& $\Omega^{(tr)}$ & 100 & 13 & 28 & 57 & 117 & 326\\
& $\Omega^{(de)}$ & 78 & 12 & 24 & 47 & 94 & 248\\\hline

\multirow{4}{*}{$T$ (min)} & $\Omega^{(day)}$ & 130 & 1.6 & 13.4 & 52 & 262 & 415\\
& $\Omega^{(te)}$ & 122 & 1.1 & 11.1 & 49 & 239 & 410\\
& $\Omega^{(tr)}$ & 85 & 4.5 & 12.8 & 33 & 102 & 339\\
& $\Omega^{(de)}$ & 40 & 4.1 & 9.8 & 23 & 56 & 140\\\hline

\multirow{4}{*}{$R$} & $\Omega^{(day)}$ & 2.12$\%$ & 0.07$\%$ & 0.30$\%$ & 0.73$\%$ & 2.33$\%$ & 8.69$\%$\\
& $\Omega^{(te)}$ & 1.78$\%$ & 0.06$\%$ & 0.25$\%$ & 0.68$\%$ & 1.89$\%$ & 7.19$\%$\\
& $\Omega^{(tr)}$ & 2.09$\%$ & 0.17$\%$ & 0.48$\%$ & 1.15$\%$ & 2.56$\%$ & 7.21$\%$\\
& $\Omega^{(de)}$ & 1.20$\%$ & 0.14$\%$ & 0.32$\%$ & 0.70$\%$ & 1.51$\%$ & 4$\%$\\\hline

\multirow{4}{*}{$r$} & $\Omega^{(day)}$ & 14.25$\%$ & 0.43$\%$ & 3.35$\%$ & 12.11$\%$ & 22.57$\%$ & 34.12$\%$\\
& $\Omega^{(te)}$ & 13.82$\%$ & 0.38$\%$ & 3.12$\%$ & 11.48$\%$ & 21.93$\%$ & 34$\%$\\
& $\Omega^{(tr)}$ & 16.01$\%$ & 3.85$\%$ & 7.98$\%$ & 15.23$\%$ & 22.76$\%$ & 32.41$\%$\\
& $\Omega^{(de)}$ & 16.99$\%$ & 2.45$\%$ & 9.36$\%$ & 17$\%$ & 24.13$\%$ & 32.97$\%$\\\hline

\multirow{4}{*}{$\sigma$} & $\Omega^{(day)}$ & 30.32$\%$ & 11.27$\%$ & 17.94$\%$ & 24.89$\%$ & 35.56$\%$ & 63.18$\%$\\
& $\Omega^{(te)}$ & 28.10$\%$ & 10.65$\%$ & 16.91$\%$ & 23.69$\%$ & 33.50$\%$ & 58.98$\%$\\
& $\Omega^{(tr)}$ & 27.89$\%$ & 10.68$\%$ & 16.89$\%$ & 23.62$\%$ & 33.388$\%$ & 58.35$\%$\\
& $\Omega^{(de)}$ & 28.92$\%$ & 10.97$\%$ & 17.38$\%$ & 24.42$\%$ & 34.50$\%$ & 60.76$\%$\\\hline

\multirow{4}{*}{$\psi$ (bp)} & $\Omega^{(day)}$ & 15.56 & 3.57 & 6.27 & 10.29 & 16.83 & 44.92\\
& $\Omega^{(te)}$ & 12.08 & 3.41 & 5.57 & 9.01 & 13.58 & 31.97\\
& $\Omega^{(tr)}$ & 11.98 & 3.57 & 6.16 & 9.73 & 14.17 & 28.85\\
& $\Omega^{(de)}$ & 1.075 & 3.40 & 5.47 & 8.64 & 12.705 & 25.75\\\hline

\end{tabular}
\caption{Statistics (mean and quantiles) on the main characteristics distribution of meta-orders for different databases.}
  \label{tab:stat:all}
\end{table}


\end{document}